\documentclass[times,twocolumn,review,longtitle]{elsarticle}

\usepackage{medima}
\usepackage{framed,multirow}

\usepackage{amssymb}
\usepackage{latexsym}

\usepackage{url}
\usepackage[dvipsnames]{xcolor}
\usepackage{hyperref}

\definecolor{newcolor}{rgb}{.8,.349,.1}

\usepackage{amsthm}
\usepackage{amsmath,amssymb,amsfonts}
\usepackage{graphicx}
\usepackage{textcomp}

\usepackage{mathrsfs}
\usepackage[section]{placeins}
\usepackage{subfigure}
\usepackage{booktabs}
\usepackage{multirow}
\usepackage{float}
\usepackage{wrapfig}
\usepackage{bbding}
\usepackage{caption}
\usepackage{subcaption}
\usepackage[dvipsnames]{xcolor}
\usepackage{algorithm}
\usepackage{algorithmic}
\newtheorem{proposition}{Proposition}

\journal{Medical Image Analysis}

\begin{document}

\verso{Huang \textit{et~al.}}

\begin{frontmatter}

\title{Enhancing Global Sensitivity and Uncertainty Quantification in Medical Image Reconstruction with Monte Carlo Arbitrary-Masked Mamba}%

\author[1,2,3]{Jiahao \snm{Huang}\corref{cor1}\fnref{fn1}}
\ead{j.huang21@imperial.ac.uk}
\cortext[cor1]{Corresponding author: Jiahao Huang, Guang Yang}
\fntext[fn1]{Co-first author}
\author[4]{Liutao \snm{Yang}\fnref{fn1}}
\author[1,2,3]{Fanwen \snm{Wang}}
\author[1,2]{Yang \snm{Nan}}
\author[5]{Weiwen \snm{Wu}}
\author[6]{Chengyan \snm{Wang}}
\author[7,8]{Kuangyu \snm{Shi}}
\author[9]{Angelica I. \snm{Aviles-Rivero}}
\author[9]{Carola-Bibiane \snm{Sch{\"o}nlieb}}
\author[4]{Daoqiang \snm{Zhang}}
\author[1,2,3]{Guang \snm{Yang}\corref{cor1}}
\ead{g.yang@imperial.ac.uk}

\address[1]{Department of Bioengineering and Imperial-X, Imperial College London, London, United Kingdom}
\address[2]{National Heart and Lung Institute, Imperial College London, London, United Kingdom}
\address[3]{Cardiovascular Research Centre, Royal Brompton Hospital, London, United Kingdom}
\address[4]{College of Computer Science and Technology, Nanjing University of Aeronautics and Astronautics, Nanjing, China}
\address[5]{School of Biomedical Engineering, Shenzhen Campus of Sun Yat-sen University, Guangdong, China}
\address[6]{Human Phenome Institute, Fudan University, Shanghai, China}
\address[7]{Department of Nuclear Medicine, Inselspital, University of Bern, Bern, Switzerland}
\address[8]{Department of Informatics, Technical University of Munich, Munich, Germany}
\address[9]{Department of Applied Mathematics and Theoretical Physics, University of Cambridge, Cambridge, United Kingdom}

\received{N.A.}
\finalform{N.A.}
\accepted{N.A.}
\availableonline{N.A.}
\communicated{N.A.}

\begin{abstract}
Deep learning has been extensively applied in medical image reconstruction, where Convolutional Neural Networks (CNNs) and Vision Transformers (ViTs) represent the predominant paradigms, each possessing distinct advantages and inherent limitations: CNNs exhibit linear complexity with local sensitivity, whereas ViTs demonstrate quadratic complexity with global sensitivity. 
The emerging Mamba has shown superiority in learning visual representation, which combines the advantages of linear scalability and global sensitivity.
In this study, we introduce MambaMIR, an Arbitrary-Masked Mamba-based model with wavelet decomposition for joint medical image reconstruction and uncertainty estimation.
A novel Arbitrary Scan Masking (ASM) mechanism ``masks out'' redundant information to introduce randomness for further uncertainty estimation. Compared to the commonly used Monte Carlo (MC) dropout, our proposed MC-ASM provides an uncertainty map without the need for hyperparameter tuning and mitigates the performance drop typically observed when applying dropout to low-level tasks. 
For further texture preservation and better perceptual quality, we employ the wavelet transformation into MambaMIR and explore its variant based on the Generative Adversarial Network, namely MambaMIR-GAN. 
Comprehensive experiments have been conducted for multiple representative medical image reconstruction tasks, demonstrating that the proposed MambaMIR and MambaMIR-GAN outperform other baseline and state-of-the-art methods in different reconstruction tasks, where MambaMIR achieves the best reconstruction fidelity and MambaMIR-GAN has the best perceptual quality. 
In addition, our MC-ASM provides uncertainty maps as an additional tool for clinicians, while mitigating the typical performance drop caused by the commonly used dropout.
\end{abstract}


\end{frontmatter}

\section{Introduction}
\label{sec:introduction}

Medical imaging reconstruction stands as one of the most fundamental and pivotal components of medical imaging. High-quality and high-fidelity reconstructed medical images ensure the precision and effectiveness of subsequent disease diagnosis and treatment planning, thus reducing potential risks to patient health~\citep{Wang2020Deep}. 
\begin{figure*}[t]
    \centering
    \includegraphics[width=0.9\linewidth]{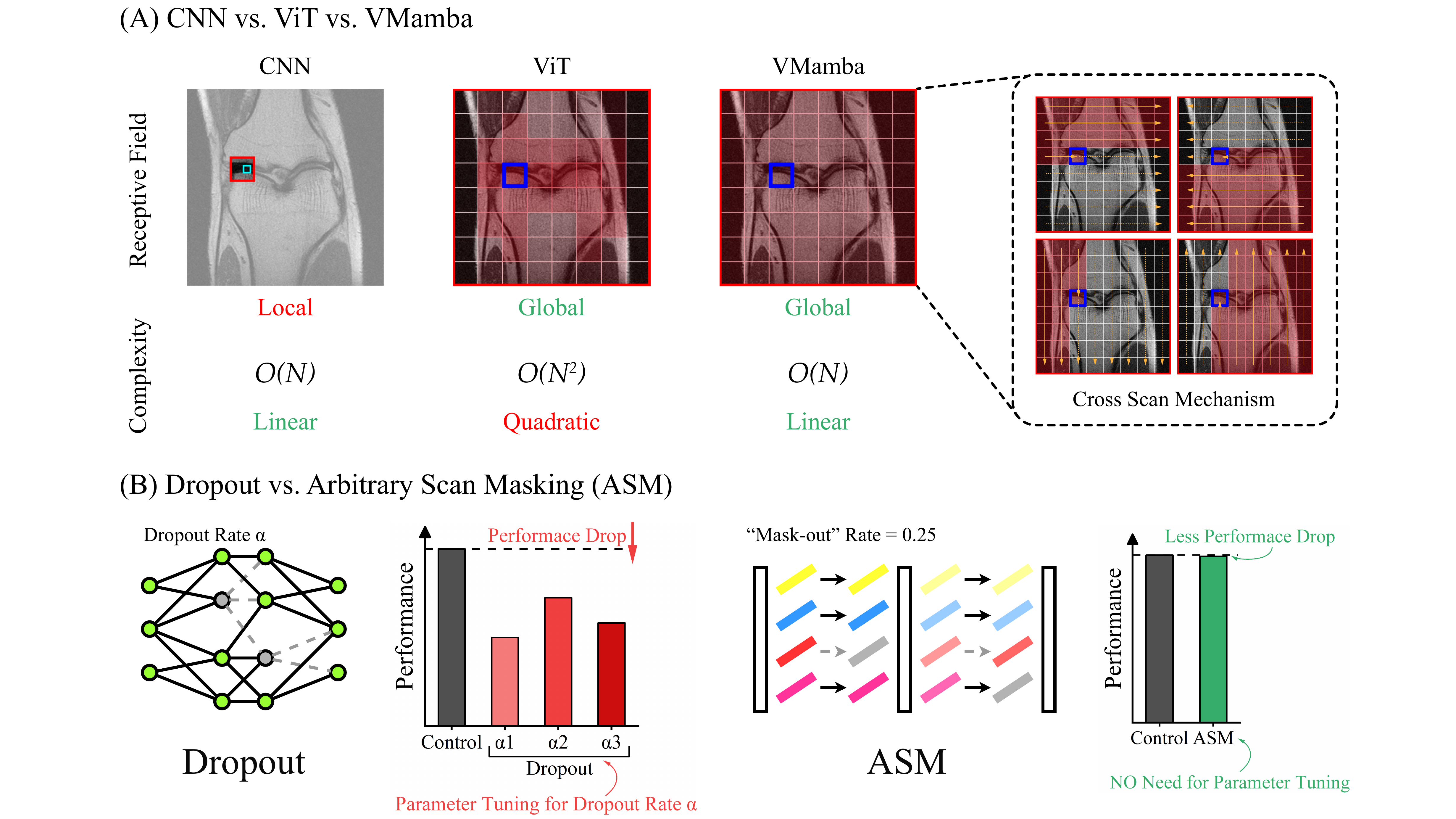}
    \caption{
    (A) Comparison between Convolutional Neural Networks (CNNs), Vision Transformers (ViTs) and VMamba. 
    CNNs and ViTs represent two predominant paradigms, each possessing distinct advantages and inherent limitations: CNNs exhibit linear complexity with local sensitivity, whereas ViTs demonstrate quadratic complexity with global sensitivity. The emerging VMamba~\citep{Liu2024VMamba} has shown superiority in computer vision tasks, combining the advantages of linear scalability and global sensitivity.
    (B) Comparison between dropout and the proposed Arbitrary Scan Masking (ASM) mechanism. 
    Dropout requires careful hyperparameter tuning (dropout rate) and typically leads to a performance drop in low-level tasks, despite its ability to mitigate overfitting in high-level tasks. 
    The proposed ASM mechanism presents a superior alternative to dropout. Instead of randomly ``dropping'' some activations that may be essential for the final outcome, our ASM strategically ``masks out'' a part of redundant information during training and inference stage.
    }
    \label{fig:FIG_INTRO}
\end{figure*}
Magnetic Resonance Imaging (MRI) provides high-resolution and reproducible assessments without exposure to radiation.
Fast MRI is widely utilised to produce MR images from sub-Nyquist sampled \textit{k}-space measurements, aiming to speed up the inherently slow data acquisition process and eliminate artefacts~\citep{Liang2020Deep, Hammernik2022Physics, Huang2024Data}. 
X-ray Computed Tomography (CT), while capable of producing high-quality and detailed images, involves radiation risks. Sparse-view CT (SVCT) has been developed to reduce radiation doses by using fewer projection views, albeit at the risk of introducing significant artefacts~\citep{shah2008alara, pan2009commercial}. 
Positron Emission Tomography (PET), critical for understanding metabolic and functional body processes, often requires long scan times or high doses for quality imaging, leading to discomfort and risk. To address this challenge, the development of low-dose PET (LDPET) presents a promising avenue to enhance image quality without increasing the injected doses~\citep{knopp2020advances}.

\begin{figure*}[t]
    \centering
    \includegraphics[width=0.8\linewidth]{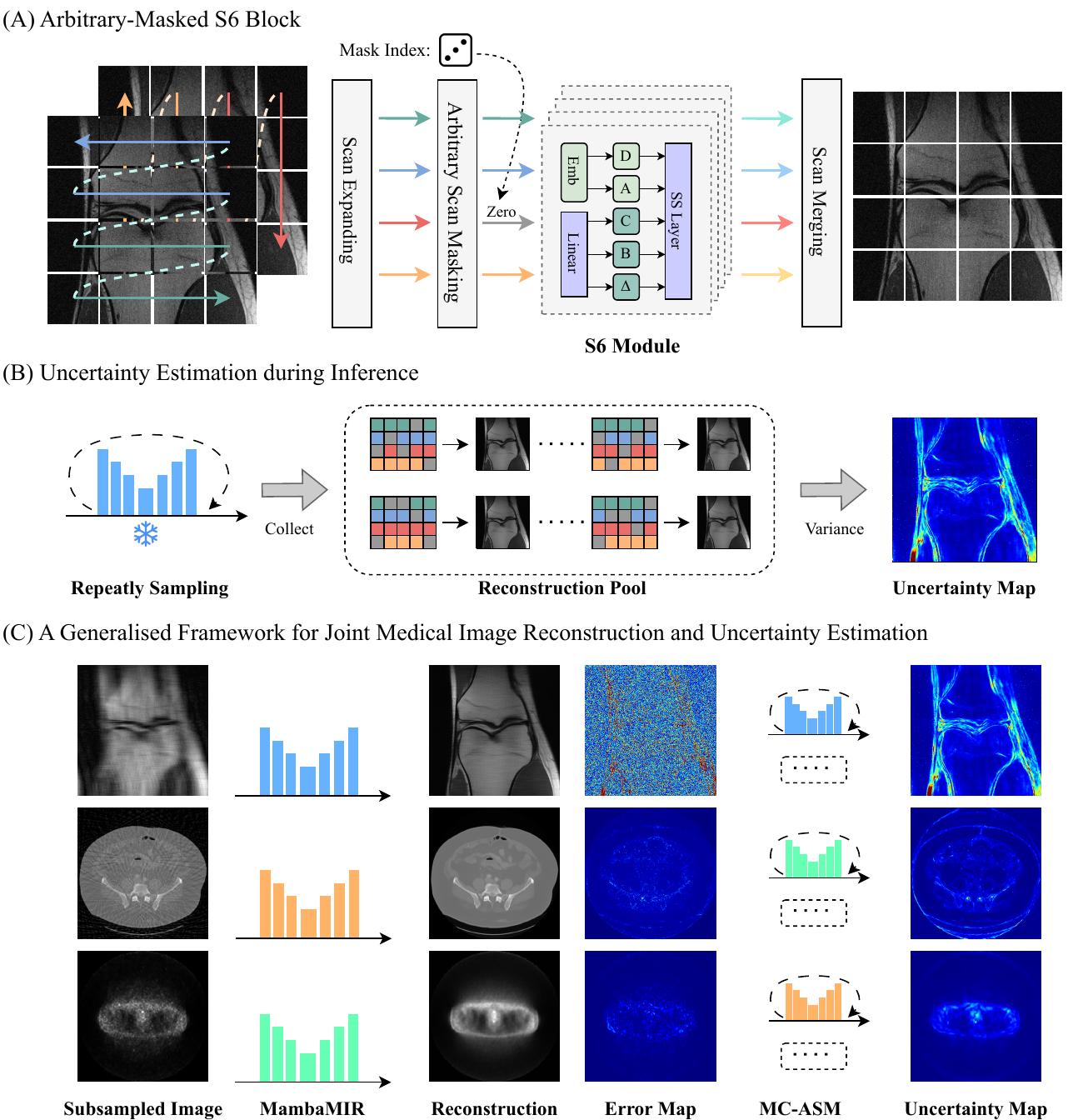}
    \caption{
    (A) The proposed Arbitrary-Masked S6 (AMS6) block. An ASM6 block includes a Scan Expanding module, an Arbitrary Scan Masking module, an S6 module, and a Scan Merging module.
    (B) Uncertainty estimation with the proposed Arbitrary Scan Masking mechanism during inference.
    (C) The framework of the proposed MambaMIR.
    }
    \label{fig:FIG_AMS6}
\end{figure*}

A key research topic and challenge for medical image reconstruction is developing an effective, efficient, and reliable reconstruction model.
Rapid advancement of artificial intelligence has propelled the development and widespread application of deep learning-based medical image reconstruction. Convolutional Neural Networks (CNNs) and Vision Transformers (ViTs)~\citep{Dosovitskiy2020ViT} represent the predominant paradigms that have achieved remarkable success in various vision tasks and are widely used in the medical imaging field. However, both CNNs and ViTs possessed their distinct advantages and inherent limitations. 

Convolutional Neural Networks excel at capturing visual features especially in identifying local patterns, taking advantage of their hierarchical architecture and inductive biases. The shared weight mechanism makes them more parameter efficient than multilayer perceptrons (MLPs). However, despite their powerful feature extraction capabilities and linear complexity, as Fig.~\ref{fig:FIG_INTRO} (A) illustrated, CNNs typically exhibit local sensitivity and a lack of long-range dependencies, limiting their ability to contextualise global features.
Vision Transformers~\citep{Dosovitskiy2020ViT}, characterised by their large receptive fields and global sensitivity, often outperform CNNs in capturing extensive contextual information. Nonetheless, as Fig.~\ref{fig:FIG_INTRO} (A) shown, their significant computational demand, due to the quadratic complexity of the self-attention mechanism~\citep{Liu2024VMamba}, limits their practicality for medical image reconstruction.
Recent Transformer-based models for medical image reconstruction have sought to mitigate these limitations by:
1) adopting a trade-off strategy that applies the self-attention mechanism within shifting windows rather than across the entire feature map~\citep{Liang2021SwinIR, Huang2022SwinMR};
2) constructing hybrid models that incorporate CNNs~\citep{Chen2021TransUNet} or Swin Transformers~\citep{Liu2021Swin}, applying ViT blocks only within deep, low-resolution latent spaces~\citep{Chen2021TransUNet, Huang2022SDAUT}.

As a powerful alternative, the emerging Mamba~\citep{Gu2023Mamba} originated from natural image processing, combines the advantages from both CNNs and ViTs. 
Mamba exhibits superior efficiency in managing long-sequence modelling due to its linear complexity and enhancement through hardware-aware optimisations. This efficiency positions Mamba as a viable contender to the prevalent self-attention mechanisms found in Transformers, particularly for tasks involving the processing of high-resolution visual data, as Fig.~\ref{fig:FIG_INTRO} (A) illustrated.

In this study, we aim to explore the potential of Mamba in the field of medical image reconstruction, and propose MambaMIR, a Mamba-based model for joint medical image reconstruction and uncertainty estimation.
According to Fig.~\ref{fig:FIG_INTRO} (A), MambaMIR has global sensitivity with linear computational complexity, especially beneficial for low-level tasks such as medical image reconstruction, which often necessitates handling long sequences (large spatial resolutions) and maintaining global sensitivity. 

In medical image reconstruction, uncertainty estimation is presented as an essential confidence assessment, providing additional information to the clinician by highlighting critical areas of concern.
Monte Carlo (MC) dropout is a commonly used uncertainty estimation method, relying on the randomness of dropout during the training and inference stages~\citep{Gal2016Dropout}. 
However, as Fig.~\ref{fig:FIG_INTRO} (B) illustrated, dropout requires careful hyperparameter tuning for the dropout rate, which is usually sensitive to the reconstruction performance.
Furthermore, dropout typically leads to a performance drop in low-level tasks like image reconstruction, despite its ability to mitigate overfitting in high-level tasks~\citep{Kong2022Reflash}.

In this study, we propose a novel Arbitrary Scan Masking (ASM) mechanism, presenting a superior alternative to dropout. Instead of randomly ``dropping'' some activations that may be essential for the final outcome, our ASM strategically ``masks out'' a part of redundant information during training and inference stage, as Fig.~\ref{fig:FIG_AMS6} (A) illustrated.
According to Fig.~\ref{fig:FIG_AMS6} (B), during the inference stage, a distribution of the reconstruction results is collected, and uncertainty maps can be produced by the variance of the resulting distribution. 

Our proposed MambaMIR is a generalised framework for joint medical image reconstruction across different image modalities and uncertainty estimation (Fig.~\ref{fig:FIG_AMS6} (C)). To the best of our knowledge, MambaMIR is the first Mamba-based model applied to medical image reconstruction. 
In addition, for the texture information preservation, wavelet decomposition mechanism is employed in the proposed MambaMIR within both the image space and the latent space.
For better perceptual quality, we further explore its variant based on the Generative Adversarial Network (GAN), namely MambaMIR-GAN. 

Experiments have shown that the proposed MambaMIR and MambaMIR-GAN outperformed other baseline and state-of-the-art (SOTA) methods for three medical image reconstruction tasks, including fast MRI, SVCT and LDPET. 
In addition, the use of the Monte Carlo-based Arbitrary Scan Masking mechanism (MC-ASM) provided uncertainty maps as an additional tool for clinicians, while mitigating the typical performance drop caused by conventional dropout.

Our main contributions are summarised as follows:
\begin{itemize}
    \item[$\bullet$] 
    We propose an innovative Mamba-based model, namely MambaMIR, for joint medical image reconstruction and uncertainty estimation (Fig.~\ref{fig:FIG_AMS6} (C)). 
    To the best of our knowledge, MambaMIR is the first Mamba-based model applied to medical image reconstruction. 
    MambaMIR inherits the advantage of global sensitivity and linear computational complexity from Mamba, especially beneficial for medical image reconstruction. 
    \item[$\bullet$]
    We design a novel ASM mechanism which introduces randomness by ``masking out'' redundant information for uncertainty estimation. 
    Compared to the commonly used MC dropout, our proposed MC-ASM provides an uncertainty map without the need for hyperparameter tuning and mitigates the performance drop typically observed when applying dropout to low-level tasks. 
    \item[$\bullet$]
    We employ the wavelet transformation in our proposed MambaMIR and explore its GAN-based variant, for further texture preservation and better perceptual quality.
    \item[$\bullet$]
    Experiments have shown that MambaMIR achieved SOTA for three medical image reconstruction tasks, including fast MRI, SVCT and LDPET.
    In addition, our proposed MC-ASM provided uncertainty maps as an additional tool for clinicians while mitigating the typical performance drop caused by the conventional dropout.
\end{itemize}

\section{Related Work}
\label{sec:related_work}

\subsection{Deep Learning-based Medical Image Reconstruction}
Deep Learning-based medical image reconstruction has witnessed significant advancements in recent years, with a diverse range of models and methodologies developed, which can be broadly categorised into three main paradigms: enhancement-based methods, unrolling-based methods and generative model-based methods~\citep{Hammernik2022Physics}.

Enhancement-based methods, such as CNNs and Transformers, represent a data-driven approach in medical imaging by mapping subsampled data to the fully-sampled equivalents in an end-to-end style. 
This methodology sidesteps the traditional requirement for explicit modelling of acquisition physics, enabling a more streamlined and efficient reconstruction process. 
The effectiveness of this strategy is underscored by its application across a diverse spectrum of imaging modalities, including MRI~\citep{Hyun2018Deep, Huang2022SwinMR}, SVCT~\citep{xia2021magic, yang2022low, yang2022learning}, and PET~\citep{gong2018pet}, where it has consistently demonstrated its ability to accurately reconstruct images by learning a non-linear mapping. 

Unrolling-based methods represent another innovative approach by integrating trainable parameters and neural networks into unrolled iterative reconstruction algorithms, such as those based on the Alternating Direction Method of Multipliers (ADMM). 
In doing so, these models can enforce data consistency more effectively by incorporating elements of the physical model directly into the reconstruction process. The ADMM-Net, for example, exemplifies the potential of unrolling-based models by learning the transformation of images and nonlinear operators, demonstrating significant improvements in the reconstruction of subsampled MRI~\citep{sun2016deep, Schlemper2017D5C5, yang2018admm}, sparse-view CT~\citep{xiang2021fista}, and low-count PET images~\citep{gong2019mapem}. 

Unlike enhancement- or unrolling-based methods, generative models focus on generating fully-sampled images from a learnt distribution, potentially circumventing the need for paired data during training. This category includes variational autoencoders, GANs, and diffusion models, each contributing uniquely to the field. 
These models leverage the power of generative algorithms to simulate realistic high-quality medical images, offering a promising avenue for image reconstruction. The flexibility and generative capacity of these models have attracted attention, recent research illustrating their potential medical image reconstruction, by embracing the inherent variability and complexity of medical data~\citep{zhao2023generative}.

\subsection{Wavelet Transformation in Image Restoration}
Discreate wavelet transform (DWT) decomposes the signal into a set of wavelets with characterisation of both frequency and location. This method is particularly impactful in the field of image reconstruction and enhancement, where it serves multiple functions. 
As a regulariser in advanced computational models, DWT guides optimisation toward convex objectives, which has been shown to be essential in unrolling networks and Plug-and-Play methods~\citep{gu2022revisiting}.
Its ability to perform high-low-frequency decomposition enhances image processing tasks such as super-resolution by effectively suppressing noise and improving convergence through multi-scale representations~\citep{yu2021wavefill}.
Moreover, integration of DWT into neural network architectures as part of downsampling and upsampling operations\citep{wu2023wavelet} enriches the network's ability to process images, reducing noise and artefacts for cleaner input. This incorporation is particularly advantageous for image inpainting and super-resolution~\citep{yu2021wavefill,wang2023disgan}.
In optimisation and texture enhancement, wavelets contribute to the design of loss functions~\citep{li2022wavelet, yang2020net}, where low-frequency loss affects holistic quality, and high-frequency loss improves perceptual quality. This dual approach ensures that images not only exhibit high fidelity, but are also visually appealing and detailed.

\subsection{Uncertainty Estimation in Medical Image Reconstruction}
Uncertainty estimation is crucial in evaluating and understanding the predictions made by deep learning models, particularly in fields like medical imaging where precise and reliable predictions are vital~\citep{zou2023review}. Bayesian Neural Networks (BNNs)~\citep{kendall2017uncertainties} present a framework for quantifying this uncertainty by placing a prior distribution over the model's weights. This approach allows BNNs to capture the inherent uncertainty in predictions, especially useful in ill-posed inverse problems where the objective is to reconstruct a fully-sampled image from limited measurements. However, the complexity of BNNs arises in their inference, as the marginal probability of the network weights cannot be directly computed.

Monte Carlo Dropout~\citep{Gal2016Dropout} offers a practical solution to this challenge by approximating variational inference. By retaining dropout during both training and inference phases, MC Dropout enables the model to sample from its posterior distribution, thus estimating uncertainty by aggregating the outcomes of multiple forward passes. This method effectively integrates with variational autoencoders~\citep{grover2019uncertainty} and diffusion models~\citep{luo2023bayesian}, enhancing their ability to quantify uncertainty. 

Ensemble methods~\citep{lambert2024trustworthy} further extend the uncertainty estimation by leveraging the diversity across multiple models or configurations to infer the uncertainty. This technique captures a wider range of behaviours and biases within the models, offering a more comprehensive view of uncertainty. 
Although more resource-intensive, ensemble approaches enhance the robustness and reliability of uncertainty estimates, making them invaluable in applications requiring high confidence in model predictions.

\subsection{State Space Model and Mamba}

State Space Models (SSMs) have emerged as a foundational framework for the analysis of sequence data, inspired by systems theory which describes a system's dynamics through its state transitions~\citep{Gu2021Efficiently}. 
SSMs are typically characterised as linear, time-invariant systems that map an input sequence $x(t) \in \mathbb{R}^L$ to an output sequence $y(t) \in \mathbb{R}^L$ through a series of hidden states $h(t) \in \mathbb{R}^N$. These models can be expressed using linear ordinary differential equations:
\begin{equation}\label{eq:SSM_C}
\begin{aligned}
h^{\prime}(t) &= \mathbf{A} h(t) + \mathbf{B} x(t), \\
y(t) &= \mathbf{C} h(t) + \mathbf{D} x(t),
\end{aligned}
\end{equation}
where $\mathbf{A} \in \mathbb{R}^{N \times N}$, $\mathbf{B} \in \mathbb{R}^{N \times 1}$, and $\mathbf{C} \in \mathbb{R}^{1 \times N}$ represent the learnable parameters, with $\mathbf{D} \in \mathbb{R}^{1}$ typically denoting a residual connection.

The structured state space sequence models (S4)~\citep{Gu2021Efficiently} and more recent Mamba~\citep{Gu2023Mamba} are based on a discretised version of these continuous models:
\begin{equation}\label{eq:SSM_D}
\begin{aligned}
h_k =\bar{\mathbf{A}} h_{k-1}+\bar{\mathbf{B}} x_k \\
y_k =\bar{\mathbf{C}} h_k+\bar{\mathbf{D}} x_k,
\end{aligned}
\end{equation} 
where $\bar{\mathbf{A}}$, $\bar{\mathbf{B}}$, $\bar{\mathbf{C}}$, $\bar{\mathbf{D}}$ are the discretised parameter, transformed by a timescale parameter $\mathbf{\Delta}$:
\begin{equation}\label{eq:SSM_D_Param}
\begin{aligned}
\bar{\mathbf{A}} = e^{\mathbf{\Delta} \mathbf{A}}, \quad
\bar{\mathbf{C}} = \mathbf{C}, \quad
\bar{\mathbf{D}} = \mathbf{D}, \\
\bar{\mathbf{B}} = (e^{\mathbf{\Delta} \mathbf{A}}-I) \mathbf{A}^{-1} \mathbf{B} \approx \mathbf{\Delta} \mathbf{B}.
\end{aligned}
\end{equation} 

Mamba introduces a novel approach in the landscape of State Space Models (SSMs) with its Selective Structured State Space Sequence Models incorporating a Scan (S6)~\citep{Gu2023Mamba}. This innovation allows for the dynamic parameterisation of the SSM, with parameters $\bar{\mathbf{B}}$, $\bar{\mathbf{C}}$, and $\mathbf{\Delta}$ being derived directly from the input data $x$, enabling an input-specific adaptation of the model. 

Mamba is considered a strong competitor to the Transformer due to its global sensitivity and linear computational complexity, and has been widely applied for various computer vision tasks.
\citet{Zhu2024VisionMamba} introduced a Mamba-based and plain (ViT-style) vision backbone, i.e., Vim, innovatively adapting Mamba for non-causal visual sequence via bi-direction scans mechanism and position embedding strategy. 
\citet{Liu2024VMamba} proposed VMamba, a hierarchical (Swin Transformer-style) vision Mamba backbone. VMamba handles non-causal visual images by the novel cross-scan mechanism, converting images into four ordered patch sequences through integrating pixels from top-left, bottom-right, top-right, and bottom-left.
\citet{Huang2024LocalMamba} developed LocalMamba with a novel window-based local scanning mechanism, effectively capturing local information while maintaining a global sensitivity. In addition, differentiable architecture search~\citep{Liu2018DARTS} was utilised for learning a optimal combination of scan modes.
Research on Mamba has also been extended to the higher-dimensional vision backbone~\citep{Li2024VideoMamba,Li2024MambaND}, where more complicated and task-specific multi-dimensional scanning mechanism were developed.
Mamba-based models have been widely utilised for various down-stream tasks across image segmentation~\citep{Ma2024UMamba,Ruan2024VMUNet,Wang2024Large}, detection~\citep{Gong2024nnMamba,Chen2024MiM} and restoration~\citep{Guo2024MambaIR,He2024PanMamba,Zheng2024FD}.

\section{Methodology}
\label{sec:methodlogy}

\subsection{Medical Image Reconstruction}
The forward acquisition process for medical images is described by:
\begin{equation}\label{eq:forward}
\begin{aligned}
\mathbf{y}=\mathbf{A} \mathbf{x} + \mathbf{n},
\end{aligned}
\end{equation}
where $\mathbf{x} \in \mathbb{C}^n$ represents the image of interest, $\mathbf{y} \in \mathbb{C}^m$ denotes the corresponding measurements, and $\mathbf{n} \in \mathbb{C}^m$ is the inevitable noise encountered during the measurement process. 

Depending on the type of medical imaging, the forward operator $\mathbf{A}$ can vary. 
For fast MRI, $\mathbf{A}$ can be a subsampled discrete Fourier transform $\mathcal{F}_{\Omega}: \mathbb{C}^n \rightarrow \mathbb{C}^m$, sampling the \textit{k}-space locations as specified by $\Omega$. 
For SVCT, $\mathbf{A}$ is represented by the Radon transform $\mathcal{R}_{\Gamma}: \mathbb{C}^n \rightarrow \mathbb{C}^m$, projecting targets into a sinogram under a selected set of imaging angles $\Gamma$.
For LDPET, $\mathbf{A}$ is represented by the detection probability matrix $\mathcal{R}_{\Delta}: \mathbb{C}^n \rightarrow \mathbb{C}^m$, detecting coincidence events of gamma photons in acquisition time $\Delta$.

Generally, the goal of the reconstruction stage is to recover the ground truth $\mathbf{x}$ from the undersampled measurements $\mathbf{y}$. This process can be formulated as an inverse problem:
\begin{equation}\label{eq:inverse}
\begin{aligned}
\hat{\mathbf{x}}=\arg \min _{\mathbf{x}} \frac{1}{2}\|\mathbf{A} \mathbf{x} - \mathbf{y}\|_2^2+\lambda \mathcal{R}(\mathbf{x}),
\end{aligned}
\end{equation} 
where $\mathcal{R}$ represents a class of regularisers, and $\lambda$ is a balancing parameter. This formulation aims to minimise the discrepancy between the measured and predicted data while incorporating regularisation to impose prior knowledge or desired properties on the solution.

\subsection{MambaMIR: Overall Architecture}

The proposed MambaMIR model adopts a U-shaped architecture, which incorporates modules for patch embedding and unembedding, along with $M$ paired encoder and decoder residual blocks, each linked by corresponding skip connections.
Within both the encoding and decoding pathways, each residual block consists of $N$ Arbitrary-Masked State Space (AMSS) blocks and includes modules for wavelet-based downsampling and upsampling. 
In the bottleneck, two Wavelet-embedded AMSS (WAMSS) blocks and a single Attention block for deep feature extraction are employed.
Furthermore, features derived from the wavelet decomposition of the initial input are integrated into the encoding pathway through the Wavelet Decomposition module.

\subsection{Wavelet Decomposition Mechanism}
To further improve reconstruction quality and texture preservation, wavelet decomposition mechanisms are incorporated into the proposed MambaMIR in various modules, inspired by~\citet{Phung2023Wavelet}. 
The 2D DWT and the inverse DWT (iDWT) are two commonly used transformations in medical images.
A 2D image $X$ can be decomposed into four subbands carrying different frequency components via DWT, while iDWT can recover the 2D image $X^\prime$ from four subbands:
\begin{equation}\label{eq:dwt}
\begin{aligned}
\{LL, HL, LH, HH\} = \operatorname{DWT}(X), \\
X^\prime = \operatorname{iDWT}(\{LL, HL, LH, HH\}),
\end{aligned}
\end{equation} 
where $LL$ indicates the low-frequency subband, $HL$, $LH$ and $HH$ are three high-frequency subbands representing vertical, horizontal and diagonal features of the original image. For conciseness, we use $H+$ to indicate these three high-frequency subbands.

\subsubsection{Wavelet-based Downsampling and Upsampling Module}
Conventionally, pooling or stride convolution are two popular methods for feature maps spatial downsampling, while deconvolution and pixel shuffling are commonly utilised for spatial upsampling. 
In terms of DWT and iDWT, one essential characteristic is their inherent spatial downsampling and upsampling mechanism, where each subband has only half of height and width compared to the original image. 
Our proposed Wavelet-based Downsampling (WDown) and Upsampling (WUp) modules leverage DWT and iDWT, inherently enabling feature downsampling and upsampling, while disentangling different frequency components.
In addition, the high-frequency skip connection is used between two paired WDown and WUp modules to transfer the high-frequency subbands ($H^+$: $HL$, $LH$, $HH$). 
The WDown module can be mathematically written as:
\begin{equation}\label{eq:wdown}
\begin{aligned}
X^{\prime} &= \operatorname{Conv}(\operatorname{GN}(X_{\text{in}})), \\
\{X_{\text{skip}}, ...\} &= \{\underbrace{LL}_{X_{\text{skip}}}, ...\} = \operatorname{DWT}(\operatorname{Conv}(X_{\text{in}})) \\
\{X^{\prime\prime}, H^{+}\} &= \{\underbrace{LL^{\prime}}_{X^{\prime\prime}}, \underbrace{HL^{\prime}, LH^{\prime}, HH^{\prime}}_{H^{+}}\} 
= \operatorname{DWT}(X^{\prime}), \\
X^{\prime\prime\prime} &= \operatorname{Conv}(\operatorname{GN}(X^{\prime\prime})), \\
X_{\text{out}} &= X^{\prime\prime\prime} + X_{\text{skip}},
\end{aligned}
\end{equation} 
where $X_{\text{in}}$ and $X_{\text{out}}$ are the input and output feature maps of the WDown module. $H^{+}$ is the high-frequency component that is transferred to the corresponding WUp module. $\operatorname{Conv}$ and $\operatorname{GN}$ indicate the convolutional layer and group normalisation. 
The WUp module can be written as:
\begin{equation}\label{eq:wup}
\begin{aligned}
X^{\prime} &= \operatorname{Conv}(\operatorname{GN}(X_{\text{in}})), \\
X_{\text{skip}} &= \operatorname{iDWT}(\{\operatorname{Conv}(X_{\text{in}}), H^{+} \}) \\
X^{\prime\prime} &= \operatorname{iDWT}(\{X^{\prime}, H^{+} \}) \\
X^{\prime\prime\prime} &= \operatorname{Conv}(\operatorname{GN}(X^{\prime\prime})) \\
X_{\text{out}} &= X^{\prime\prime\prime} + X_{\text{skip}},
\end{aligned}
\end{equation} 
where $X_{\text{in}}$ and $X_{\text{out}}$ are the input and output feature maps of the WUp module. 

\subsubsection{Wavelet Decomposition and Wavelet Information Fusion}

To better preserve texture information, additional skip connections are used to integrate wavelet-derived information from the original image into the encoder’s feature maps. In the wavelet information pathway, Wavelet Decomposition modules are utilised for wavelet information extraction and spatial downsampling to adapt the resolution of encoder feature maps, which can be expressed as:
\begin{equation}\label{eq:wavelet}
\begin{aligned}
X^{\prime} &= \{LL, HL, LH, HH \} = \operatorname{DWT}(X_{\text{in}}), \\
X_{\text{out}} &= \operatorname{Conv}(X^{\prime}), \\
\end{aligned}
\end{equation} 
where $X_{\text{in}}$ and $X_{\text{out}}$ are the input and output of the Wavelet Decomposition module. 
The wavelet-derived information is further integrated into the encoding pathway, functioning similarly to skip connections, but specifically for wavelet features.

\subsubsection{Wavelet AMSS}
In addition, two additional WAMSS blocks are applied to enhance the feature extraction from low-frequency components by AMSS blocks, while attempt to preserve more high-frequency detail via skip connections, which can be expressed as:
\begin{equation}\label{eq:wamss}
\begin{aligned}
\{X^{\prime}, H^{+}\} &= \{\underbrace{LL}_{X^{\prime}}, \underbrace{HL, LH, HH}_{H^{+}}\} = \operatorname{DWT}(X_{\text{in}}), \\
X^{\prime\prime} &= \operatorname{AMSS}(X^{\prime}), \\
X_{\text{out}} &= \operatorname{iDWT}(\{X^{\prime\prime}, H^{+}\}), \\
\end{aligned}
\end{equation} 
where $X_{\text{in}}$ and $X_{\text{out}}$ are the input and output of the WAMSS blocks. 

\subsection{AMSS Block}

The structure of the AMSS block follows the design of the Mamba block~\citep{Gu2023Mamba} and VSS block~\citep{Liu2024VMamba}.
The input of AMSS blocks goes through a layer normalisation step before being divided into two pathways. The first pathway follows a sequence of layers: a gating linear layer, a depth-wise convolution layer with a 3$\times$3 kernel, a SiLU activation function~\citep{Ramachandran2017Searching}, an Arbitrary-Masked S6 (AMS6) block and another layer normalisation layer. Meanwhile, the second pathway involves a linear layer with a SiLU activation function. The results of these two pathways are combined by multiplication and then passed through a final gating linear layer to generate the output of the AMSS block. 
The AMSS Block can be mathematically written as:
\begin{equation}\label{eq:amss}
\begin{aligned}
X^{\prime} &= \operatorname{LN}(X_{\text{in}}), \\
X^{\prime\prime} &= \operatorname{LN}(\operatorname{AMS6}(\operatorname{DWConv}(\operatorname{Linear}(X^{\prime})))), \\
X_{\text{gate}} &= \operatorname{Linear}(X_{\text{in}}), \\
X_{\text{out}} &= \operatorname{Linear}(X_{\text{gate}} \odot X^{\prime\prime}) + X_{\text{in}}, \\
\end{aligned}
\end{equation} 
where $X_{\text{in}}$ and $X_{\text{out}}$ are the input and output of the AMSS blocks. $\operatorname{DWConv}$ is the depth-wise convolution layer, $\operatorname{LN}$ is the layer normalisation layer, and $\odot$ is the Hadamard production.

\subsection{Monte Carlo-based Arbitrary Scan Masking}

\subsubsection{Arbitrary-Masked S6 Block}

A challenge arises when using Mamba to process vision data. S6 inherently processes data in an ordered sequential style, where information integration is limited to data that has been sequentially processed. This characteristic aligns well with temporal natural language processing tasks, however, posing challenges for computer vision tasks where data is not strictly sequential.
Existing methods have been developed to mitigate this challenge by re-ordering the visual sequence by various directions, meanwhile leading to redundancy in sequence information~\citep{Zhu2024VisionMamba, Liu2024VMamba}.

In this study, we introduce the AMS6 block, a novel component aimed at improving the performance of State Space Models in processing visual data, as Fig.~\ref{fig:FIG_AMS6} illustrated.
Our proposed AMS6 block incorporates the cross-scan mechanism~\citep{Liu2024VMamba}, to adapt Mamba to medical image data, while leveraging the inherent redundancy for uncertainty estimation.
The AMS6 block includes four key modules: the Scan Expanding module, the ASM module, the S6 module, and the Scan Merging module. 
The pseudo-code is presented in Algorithm~\ref{alg:ams6}.
\begin{algorithm}[ht]
    \renewcommand{\algorithmicrequire}{\textbf{Input:}}
    \renewcommand{\algorithmicensure}{\textbf{Output:}}
    \caption{Arbitrary-Masked S6 Block} 
    \label{alg:ams6} 
    \setlength{\tabcolsep}{4mm}{
    \begin{algorithmic}
        \REQUIRE $X$ \quad \textcolor{OliveGreen}{\# feature map, X.shape: $(B, C, H, W)$;}
        \STATE \textcolor{OliveGreen}{\# $B$: batch size, $C$: channel, $H$: height, $W$: width;}
        \STATE \quad
        
        \STATE \textcolor{OliveGreen}{\# Scan Expanding Module}
        \STATE $Xs \gets \operatorname{ScanExpand}(X)$ \quad \textcolor{OliveGreen}{\# xs.shape: $(B, 4, C, H, W)$}
        \STATE \quad
        
        \STATE \textcolor{OliveGreen}{\# Arbitrary Scan Masking Mechanism}
        \STATE $s \gets \operatorname{random\_int}({0, 4})$
        \STATE $Xs\_m \gets Xs$ \quad \textcolor{OliveGreen}{\# Xs\_m.shape: $(B, 4, C, H, W)$;}
        \STATE $Xs\_m[:, s, ...] \gets \operatorname{zeros\_like}(Xs[:, s, ...]) $ 
        \STATE \quad

        \STATE \textcolor{OliveGreen}{\# S6 Module}
        \STATE $Ys \gets \operatorname{S6}(Xs\_m)$ \quad \textcolor{OliveGreen}{\# Ys.shape: $(B, 4, C, H, W)$;}
        \STATE \quad
        
        \STATE \textcolor{OliveGreen}{\# Scan Merging Module}
        \STATE $Y \gets \operatorname{ScanMerge}(Ys)$ \quad \textcolor{OliveGreen}{\# Y.shape: $(B, C, H, W)$;}
        \STATE \quad
        
        \ENSURE $Y$ \quad \textcolor{OliveGreen}{\# feature map, Y.shape: $(B, C, H, W)$;}
    \end{algorithmic}}
\end{algorithm}

The Scan Expanding module extends image patches across rows or columns, beginning from the upper-left or lower-right corner, transforming a single image into four distinct ordered sequences, as Fig.~\ref{fig:FIG_AMS6} (A) illustrated.
The expanding process results in a $4\times$ expansion of an image, making it redundant because all scans contain identical information, with the only variation being the direction of the scan.

Randomness is introduced in the ASM module via arbitrary scan masking, which takes advantage of the redundancy of scans. This is achieved by nullifying the pixels in one out of the four scans randomly chosen, selectively masking-out information while keeping the original matrix shape unchanged. In this way, our ASM mitigates the performance drop typically observed when applying dropout to low-level tasks.

The S6 module is the core component of the AMS6 block, responsible for processing scan-expanded sequences. 
Subsequently, these processed scans are merged and re-organised into their original patch form by the Scan Merging module. 

The integration within the AMS6 block enhances the image reconstruction process meanwhile introducing randomness for the further uncertainty estimation via MC-ASM mechanism.

\subsubsection{Monte Carlo-based Arbitrary Scan Masking Mechanism}
Our proposed MC-ASM achieves uncertainty estimation by producing a distribution of predictions (reconstruction) from a single input (subsampled images) during inference stage, utilising randomness from ASM module, inspired by MC dropout~\citep{Gal2016Dropout}. 

From a Bayesian perspective, the arbitrary masking procedure can be interpreted as a variational inference method to approximate the posterior distribution of a model's weights $p(\theta | \mathcal{D})$, which accounts for uncertainty after observing the data $\mathcal{D}$, as expressed in the Bayesian theorem:
\begin{equation}\label{eq:mcdropout1}
\begin{aligned}
p(\theta | \mathcal{D}) \propto p(\mathcal{D} | \theta) p(\theta),
\end{aligned}
\end{equation} 
where $p(\theta)$ encodes our prior knowledge about the parameters before observing the data, and $p(\mathcal{D} | \theta)$ represents the likelihood of the data given the parameters. 
The posterior predictive distribution is integral to Bayesian predictive modelling:
\begin{equation}\label{eq:mcdropout2}
\begin{aligned}
p(\mathbf{Y} | \mathbf{X}, \mathcal{D}) = \int p(\mathbf{Y} | \mathbf{X}, \theta) p(\theta | \mathcal{D}) d\theta,
\end{aligned}
\end{equation} 
where $\mathbf{X}$ is the input and $\mathbf{Y}$ is the resulting output.

In practice, MC-ASM approximates the posterior predictive distribution $p(\mathbf{Y} | \mathbf{X}, \mathcal{D})$ by repeatedly sampling from the model with ASM at the inference stage and obtaining a set of outputs $\{\mathbf{Y}_1, \mathbf{Y}_2, \ldots, \mathbf{Y}_N\}$, corresponding to a diverse set of sub-models $\{\theta_1, \theta_2, \ldots, \theta_N\}$. The variance of these outputs can be used to quantify the model's predictive uncertainty:
\begin{equation}\label{eq:mcdropout3}
\begin{aligned}
\mathbb{E}[\mathbf{Y} | \mathbf{X}, \mathcal{D}] &\approx \frac{1}{N} \sum_{i=1}^{N} \mathbf{Y}_i, \\ 
\operatorname{Var}[\mathbf{Y} | \mathbf{X}, \mathcal{D}] &\approx \frac{1}{N} \sum_{i=1}^{N}
(\mathbf{Y}_i - \mathbb{E}[\mathbf{Y} | \mathbf{X}, \mathcal{D}])^2.
\end{aligned}
\end{equation} 
This formulation allows for an empirical estimate of the epistemic uncertainty associated with the predictions, providing additional information for clinicians.

\subsubsection{Arbitrary Scan Masking is a Special Case of Dropout}
We next provide the theoretical guarantee for the proposed MC-ASM, demonstrating that the ASM mechanism can be regarded as a special case of dropout. 
\begin{proposition}
\label{main_proposition}
The Arbitrary Scan Masking mechanism can be regarded as a special case of dropout.
\end{proposition}
\begin{proof}
Dropout is a common regularisation technique widely used in deep learning, which randomly and temporarily removes a fraction of neurons and their connections in certain layers according to a predetermined probability.
In the actual implementation, dropout is typically applied to the output of a layer within a neural network:
\begin{equation}\label{eq:p_dropout1}
\begin{aligned}
X_{\operatorname{Dropout}}^{(l)} &= X^{(l)} \odot D^{(l)},
\end{aligned}
\end{equation} 
where $X^{(l)}$ is the output of $l^\text{th}$ layer in the neural network with a shape of $(B, C, H, W)$ and $X_{\operatorname{Dropout}}^{(l)}$ is the result after dropout operation. 
$D^{(l)}$ is a mask for dropout with the same shape as $X_{\operatorname{Dropout}}^{(l)}$. $\odot$ is the Hadamard production.
For simplicity, typically each component of the dropout mask $D_{b,c,h,w}^{(l)}$ independently follows a Bernoulli distribution:
\begin{equation}\label{eq:p_dropout2}
\begin{aligned}
D_{b,c,h,w}^{(l)} &\sim \operatorname{Bernoulli} (1 - p),
\end{aligned}
\end{equation} 
where $p$ is a predetermined dropout rate.

For the proposed AMS6 block, the ASM mechanism can be written as:
\begin{equation}\label{eq:p_ams6_1}
\begin{aligned}
X_{\operatorname{AM}}^{(l)} &= X^{(l)} \odot M^{(l)},
\end{aligned}
\end{equation} 
where $X^{(l)}$ is the expanded feature map after Scan Expanding module at $l^\text{th}$ AMS6 block, with a shape of $(B, 4, C, H, W)$, and $X_{\operatorname{AM}}^{(l)}$ is the result after ASM mechanism. 
$M^{(l)}$ is a mask with the same shape of $X_{\operatorname{AM}}^{(l)}$, which can be written as:
\begin{equation}\label{eq:p_ams6_2}
\begin{aligned}
M_{b,s,c,h,w}^{(l)} &= 
    \begin{cases} 
    0, & s = s', \\
    1, & s \neq s',
    \end{cases} \\
s' &\sim \operatorname{Uniform}\{0,1,2,3\},
\end{aligned}
\end{equation} 
where the index $s'$ of masked scan is randomly selected in a uniform distribution of $\{0,1,2,3\}$, and components belong to the corresponding scan are masked-out with zero value.

We demonstrate that the dropout and our ASM mechanism are inherently sharing the same mathematical form, and our ASM mechanism is a specially case of dropout with different ``dropout selection'' mechanisms. 
The minimum selection unit for typical dropout is a single matrix element, and the probability of each element is independent and follows a Bernoulli distribution controlled by a dropout rate of $p$.
However, the minimum selection unit for the ASM mechanism is a scan, and the probability of each scan is not independent.
\end{proof}

\subsection{Optimisation Scheme}

Our proposed MambaMIR can be trained and tested in an end-to-end style, represented as $\hat{\mathbf{x}}_u = \operatorname{MambaMIR}(\mathbf{x}_u)$, where $\mathbf{x}_u$ and $\hat{\mathbf{x}}_u$ denote the subsampled input and the resulting reconstruction. 

A hybrid loss $\mathcal{L}_{\mathrm{Tot}}(\theta)$ is employed for model training, involving a Charbonnier loss~\citep{Lai2019Fast} in both the image and frequency domains, which are presented as $\mathcal{L}_{\mathrm{img}}(\theta)$ and $\mathcal{L}_{\mathrm{freq}}(\theta)$, respectively. 
For better perceptual reconstruction quality, we pose a $l_{1}$ restriction on the latent space by a pre-trained VGG model $f_{\mathrm{VGG}}(\cdot)$~\citep{Simonyan2014VGG} and have $\mathcal{L}_{\mathrm{perc}}(\theta)$.
These loss functions are defined as:
\begin{equation}\label{eq:loss}
\begin{aligned}
\mathop{\text{min}}\limits_{\theta} \mathcal{L}_{\mathrm{img}}(\theta) 
&=\sqrt{\mid\mid \mathbf{x} - \hat{\mathbf{x}}_u \mid\mid^2_2 + \epsilon^2}, \\
\mathop{\text{min}}\limits_{\theta} \mathcal{L}_{\mathrm{freq}}(\theta) 
&=\sqrt{\mid\mid \mathcal{F}\mathbf{x} - \mathcal{F} \hat{\mathbf{x}}_u \mid\mid^2_2 + \epsilon^2}, \\
\mathop{\text{min}}\limits_{\theta} \mathcal{L}_{\mathrm{perc}}(\theta) 
&= \mid\mid f_{\mathrm{VGG}}(\mathbf{x}) - f_{\mathrm{VGG}}(\hat{\mathbf{x}}_u) \mid\mid_1, \\
\mathcal{L}_{\mathrm{Tot}}(\theta)
& = \alpha \mathcal{L}_{\mathrm{img}}(\theta)
+ \beta \mathcal{L}_{\mathrm{trans}}(\theta)
+ \gamma \mathcal{L}_{\mathrm{perc}}(\theta),
\end{aligned}
\end{equation} 
where $\mathbf{x}$ is the ground truth. $\epsilon$ in the Charbonnier loss is empirically set to $10^{-9}$. The trainable network parameter of the proposed MambaMIR is denoted as $\theta$. $\mathcal{F}$ represents the Discrete Fourier Transformation. $\alpha$, $\beta$ and $\gamma$ are parameters balancing different losses. 

For the GAN-based variant, i.e., MambaMIR-GAN, our MambaMIR is applied as the generator $G_{\theta_{G}}$ parameterised by $\theta_{G}$ and a U-Net discriminator~\citep{Schonfeld2020UNet} parameterised by $\theta_{D}$ for adversarial training.
The adversarial loss $\mathcal{L}_{\mathrm{adv}}(\theta_{G}, \theta_{D})$ and the total loss for MambaMIR-GAN are written as:
\begin{equation}\label{eq:loss_gan}
\begin{aligned}
\mathop{\text{min}}\limits_{\theta_{G}} 
\mathop{\text{max}}\limits_{\theta_{D}}
\mathcal{L}_{\mathrm{adv}}(\theta_{G}, \theta_{D}) 
&=\mathbb{E}_{\mathbf{x} \sim p_{\mathrm{t}}(\mathbf{x})}
[\mathop{\text{log}} D_{\theta_{D}}(\mathbf{x})] \\
&- \mathbb{E}_{\hat{\mathbf{x}}_u \sim p_{\mathrm{u}}(\hat{\mathbf{x}}_u)}
[\mathop{\text{log}} D_{\theta_{D}}(\hat{\mathbf{x}}_u)], \\
\mathcal{L}_{\mathrm{Tot-GAN}}(\theta_{G}, \theta_{D})
&= \mathcal{L}_{\mathrm{Tot}}(\theta_{G})
+ \eta \mathcal{L}_{\mathrm{adv}}(\theta_{G}, \theta_{D}),
\end{aligned}
\end{equation} 
where $\eta$ is the weighting parameter.

\section{Experiments}
\label{sec:experiments}

\subsection{Dataset}

In this work, we used the
FastMRI knee dataset~\citep{Zbontar_2018_fastMRI} and Stanford knee MRI dataset (SKMTEA)~\citep{Desai_2023_SKMTEA} for fast MRI reconstruction, 
two distinct anatomical subsets from Low-Dose CT Image and Projection Datasets~\citep{moen2021low} for SVCT reconstruction,
along with an in-house PET datasets for LDPET reconstruction.

\subsubsection{MRI: FastMRI}
For FastMRI dataset~\citep{Zbontar_2018_fastMRI}, we used 584 three-dimensional (3D) proton density weighted knee MRI scans with available ground truth, which were acquired with 15 coils without fat suppression. Within each case, 20 slices of 2D coronal-view complex-value images near the centre were utilised and centre-cropped to a resolution of $320 \times 320$ in the image space.
We randomly divided all 2D slices following a ratio of 7:1:2, into training set (420 cases), validation set (64 cases) and testing set (100 cases). The officially emulated single-coil data were applied as complex-value ground truth.

\subsubsection{MRI: SKMTEA}
For SKM-TEA dataset~\citep{Desai_2023_SKMTEA}, 155 scans of 3D, quantitative double-echo steady-state knee MRI scans with available ground truth were applied in the experiments section.
To avoid including very noisy or void slices, 100 sagittal-view 2D single-channel complex-value echo \#1 slices were chosen for each case. All slices were centre-cropped to $512 \times 512$ in the image space. 
We split 155 cases into training set (86 cases), validation set (33 cases) and testing set (36 cases), following the official dataset splits.

\subsubsection{CT: Low-Dose CT Image and Projection Datasets}
For Low-Dose CT Image and Projection Datasets~\citep{moen2021low}, two subsets including a) chest and b) abdomen were applied. The chest subset consisted of low-dose non-contrast scans aimed at screening high-risk patients for pulmonary nodules, while the abdomen subset consisted of contrast-enhanced CT scans used to detect metastatic liver lesions. 
Each subset included 40 cases. We spilt scans in each subset into training set (32 cases in chest scans; 32 cases in abdomen scans) and testing set (8 cases in chest scans; 8 cases in abdomen scans). 
Sparse-view sinograms were generated in a fan-beam CT geometry, with 60 projection views and 736 detectors. The source-to-detector distance was set to 1000 mm, and the source-to-rotation-centre distance was 512 mm. The reconstructed image resolution was $512 \times 512$ pixels.

\subsubsection{PET: Low-Dose PET}
For Low-Dose PET reconstruction, we used an in-house PET datasets contains 103 subjects of whole-body imaging. 
The low-dose data were obtained by resampling the original data, simulating various acquisition times. The dose reduction factor (DRF) quantifies the data acquired within a reduced time window, reflecting the degree of radiation dose. 
 The dataset were divided into training set (82 cases) and testing set (20 cases). The resolution of PET images is $192 \times 192$ pixels.

\subsection{Implementation Details and Evaluation Metrics}
For the network hyperparameter, we applied 4 residual blocks symmetrically in encoder and decoder paths, where each residual block consists of 2 AMSS blocks. The basic embedding channel is 180, with a multiplication factor $\{1, 2, 2, 2\}$ from shallow to deep. 
We trained our proposed MambaMIR and MambaMIR-GAN on two NVIDIA A100 (80GB) and tested them on an NVIDIA RTX 3090 GPU (24GB). 

Both MambaMIR and MambaMIR-GAN were trained using Adam optimiser for 100,000 gradient steps with a batch size of 8. The balancing parameters $\alpha$, $\beta$, $\gamma$ and $\eta$ were set to 15, 0.1, 0.0025 and 0.1. The initial learning rate was set to 0.0002, with a decay rate of 0.5 every 20,000 steps after 50,000$^\text{th}$ step.
Specifically for MambaMIR-GAN, we applied MambaMIR as generator and applied a U-Net-based discriminator~\citep{Schonfeld2020UNet} for adversarial training.

Three metrics including Peak Signal-to-Noise Ratio (PSNR), Structural Similarity Index Measure (SSIM) and Learned Perceptual Image Patch Similarity (LPIPS)~\citep{Zhang2018LPIPS} were applied for reconstruction quality assessment.

\begin{table*}[!t]
  \caption{
  Quantitative results for comparisonal studies for fast MRI, sparse-view CT (SVCT) and low-dose PET (LDPET) reconstruction.
  For fast MRI, experiments are performed on FastMRI at accelerate factor (AF) $\times 4$, $\times 8$, as well as SKMTEA at AF $\times 8$, $\times 16$. 
  For SVCT, experiments are conducted on the abdomen and chest subsets from Low-Dose CT Image and Projection Datasets.
  For LDPET, experiments are conducted on in-house dataset with dose reduction factor (DRF) $\times 3$, $\times 6$. 
  The best scores are indicated by \textbf{bold}. 
  $^{\star}$ denotes results that are significantly different from the best results by the Mann-Whitney Test ($p<0.05$).
  }
  \centering
  \resizebox{0.85\textwidth}{!}{
    \begin{tabular}{ccccccc}
    \toprule
    \multirow{2}[4]{*}{Method} & \multicolumn{3}{|c}{AF $\times$ 4} & \multicolumn{3}{|c}{AF $\times$ 8} \\
\cmidrule{2-7}          & \multicolumn{1}{|c}{SSIM $\uparrow$}  & PSNR $\uparrow$  & LPIPS $\downarrow$ & \multicolumn{1}{|c}{SSIM $\uparrow$}  & PSNR $\uparrow$  & LPIPS $\downarrow$ \\
    \midrule
    ZF    & 0.609 (0.088)$^*$ & 26.13 (2.06)$^*$ & 0.338 (0.050)$^*$ & 0.482 (0.098)$^*$ & 22.75 (1.73)$^*$ & 0.504 (0.058)$^*$ \\
    D5C5  & 0.671 (0.101)$^*$ & 28.85 (2.73)$^*$ & 0.168 (0.034)$^*$ & 0.548 (0.111)$^*$ & 25.99 (2.14)$^*$ & 0.292 (0.039)$^*$ \\
    DAGAN & 0.651 (0.093)$^*$ & 27.53 (2.05)$^*$ & 0.216 (0.048)$^*$ & 0.530 (0.106)$^*$ & 25.19 (2.21)$^*$ & 0.262 (0.043)$^*$ \\
    MWCNN & 0.696 (0.099) & 29.47 (2.72) & 0.179 (0.048)$^*$ & 0.566 (0.122)$^*$ & 26.99 (2.51)$^*$ & 0.261 (0.052)$^*$ \\
    SwinMR & 0.680 (0.103)$^*$ & 29.27 (2.87)$^*$ & 0.160 (0.037)$^*$ & 0.568 (0.116)$^*$ & 26.98 (2.47)$^*$ & 0.254 (0.043)$^*$ \\
    STGAN & 0.686 (0.098)$^*$ & 28.94 (2.65)$^*$ & 0.111 (0.034)$^*$ & 0.594 (0.105)$^*$ & 26.90 (2.31)$^*$ & \textbf{0.155 (0.040)} \\
    DiffuseRecon & 0.686 (0.103)$^*$ & 29.31 (2.73)$^*$ & 0.180 (0.030)$^*$ & 0.581 (0.118)$^*$ & 27.40 (2.40)$^*$ & 0.287 (0.038)$^*$ \\
    \midrule
    MambaMIR & 0.699 (0.103) & \textbf{29.61 (2.86)} & 0.172 (0.051)$^*$ & 0.598 (0.113)$^*$ & \textbf{27.53 (2.52)} & 0.259 (0.061)$^*$ \\
    MambaMIR-GAN & \textbf{0.703 (0.101)} & 29.36 (2.77)$^*$ & \textbf{0.109 (0.037)} & \textbf{0.617 (0.108)} & 27.33 (2.38)$^*$ & \textbf{0.155 (0.044)} \\
    \midrule
    \multirow{2}[4]{*}{Method} & \multicolumn{3}{|c}{AF $\times$ 8} & \multicolumn{3}{|c}{AF $\times$ 16} \\
\cmidrule{2-7}          & \multicolumn{1}{|c}{SSIM $\uparrow$}  & PSNR $\uparrow$  & LPIPS $\downarrow$ & \multicolumn{1}{|c}{SSIM $\uparrow$}  & PSNR $\uparrow$  & LPIPS $\downarrow$ \\
    \midrule
    ZF    & 0.529 (0.047)$^*$ & 23.25 (1.11)$^*$ & 0.462 (0.025)$^*$ & 0.473 (0.045)$^*$ & 21.28 (1.130)$^*$ & 0.555 (0.027)$^*$ \\
    D5C5  & 0.623 (0.042)$^*$ & 26.30 (1.18)$^*$ & 0.242 (0.030)$^*$ & 0.550 (0.049)$^*$ & 23.29 (1.123)$^*$ & 0.371 (0.037)$^*$ \\
    DAGAN & 0.555 (0.045)$^*$ & 24.55 (1.06)$^*$ & 0.289 (0.038)$^*$ & 0.479 (0.048)$^*$ & 22.21 (1.052)$^*$ & 0.375 (0.037)$^*$ \\
    MWCNN & 0.579 (0.041)$^*$ & 26.83 (1.21)$^*$ & 0.247 (0.026)$^*$ & 0.489 (0.045)$^*$ & 24.51 (1.23)$^*$ & 0.329 (0.034)$^*$ \\    
    SwinMR & 0.601 (0.039)$^*$ & 27.17 (1.24)$^*$ & 0.230 (0.025)$^*$ & 0.497 (0.044)$^*$ & 24.46 (1.232)$^*$ & 0.318 (0.033)$^*$ \\
    STGAN & 0.648 (0.044)$^*$ & 26.89 (1.17)$^*$ & 0.138 (0.028) & 0.565 (0.051)$^*$ & 24.43 (1.190)$^*$ & 0.209 (0.038)$^*$ \\
    DiffuseRecon & 0.584 (0.032)$^*$ & 26.27 (1.09)$^*$ & 0.202 (0.025)$^*$ & 0.478 (0.036)$^*$ & 22.92 (1.103)$^*$ & 0.325 (0.034)$^*$ \\
    \midrule
    MambaMIR & 0.620 (0.038)$^*$ & \textbf{27.43 (1.25)} & 0.237 (0.027)$^*$ & 0.514 (0.043)$^*$ & \textbf{25.07 (1.279)} & 0.312 (0.035)$^*$ \\
    MambaMIR-GAN & \textbf{0.656 (0.046)} & 27.00 (1.20)$^*$ & \textbf{0.136 (0.026)} & \textbf{0.583 (0.053)} & 24.85 (1.239) & \textbf{0.198 (0.037)} \\
    \midrule
    \multirow{2}[4]{*}{Model} & \multicolumn{3}{|c}{Abdomen} & \multicolumn{3}{|c}{Chest} \\
\cmidrule{2-7}          & \multicolumn{1}{|c}{SSIM $\uparrow$}  & PSNR $\uparrow$  & LPIPS $\downarrow$ & \multicolumn{1}{|c}{SSIM $\uparrow$}  & PSNR $\uparrow$  & LPIPS $\downarrow$ \\
    \midrule
    FBP   & 0.716 (0.041)$^*$ & 31.87 (1.25)$^*$ & 0.410 (0.035)$^*$ & 0.550 (0.036)$^*$ & 28.52 (1.02)$^*$ & 0.440 (0.037)$^*$ \\
    Inter & 0.947 (0.008)$^*$ & 40.95 (0.97)$^*$ & 0.076 (0.018)$^*$ & 0.847 (0.025)$^*$ & 35.56 (1.02)$^*$ & 0.155 (0.031)$^*$ \\
    DDNet & 0.941 (0.010)$^*$ & 40.35 (1.00)$^*$ & 0.090 (0.016)$^*$ & 0.835 (0.027)$^*$ & 35.41 (1.00)$^*$ & 0.154 (0.023)$^*$ \\
    FBPConv & 0.929 (0.013)$^*$ & 38.14 (1.09)$^*$ & 0.175 (0.025)$^*$ & 0.801 (0.035)$^*$ & 34.13 (0.97)$^*$ & 0.357 (0.050)$^*$ \\
    IRadonMap & 0.968 (0.007)$^*$ & 43.32 (1.10)$^*$ & 0.060 (0.015)$^*$ & 0.868 (0.029)$^*$ & 36.84 (1.16)$^*$ & 0.135 (0.025)$^*$ \\
    RegFormer & 0.966 (0.007)$^*$ & 42.64 (1.13)$^*$ & 0.077 (0.021)$^*$ & 0.850 (0.035)$^*$ & 36.15 (1.13)$^*$ & 0.226 (0.062)$^*$ \\
    \midrule
    MambaMIR & \textbf{0.983 (0.005)} & \textbf{45.72 (1.31)} & 0.037 (0.016)$^*$ & \textbf{0.874 (0.042)} & \textbf{37.18 (1.52)} & 0.160 (0.031)$^*$ \\
    MambaMIR-GAN & 0.977 (0.006)$^*$ & 44.57 (1.29)$^*$ & \textbf{0.024 (0.008)} & 0.854 (0.045)$^*$ & 36.33 (1.56)$^*$ & \textbf{0.049 (0.010)} \\
    \midrule
    \multirow{2}[4]{*}{Model} & \multicolumn{3}{|c}{DRF $\times 3$} & \multicolumn{3}{|c}{DRF $\times 6$} \\
\cmidrule{2-7}          & \multicolumn{1}{|c}{SSIM $\uparrow$}  & PSNR $\uparrow$  & LPIPS $\downarrow$ & \multicolumn{1}{|c}{SSIM $\uparrow$}  & PSNR $\uparrow$  & LPIPS $\downarrow$ \\
    \midrule
    Subsampled & 0.963 (0.034)$^*$ & 38.61 (4.98)$^*$ & 0.027 (0.024)$^*$ & 0.923 (0.056)$^*$ & 34.59 (4.99)$^*$ & 0.068 (0.040)$^*$ \\
    U-Net & 0.951 (0.053)$^*$ & 39.26 (4.74)$^*$ & 0.027 (0.070)$^*$ & 0.965 (0.029)$^*$ & 37.99 (4.21)$^*$ & 0.027 (0.030)$^*$ \\
    REDCNN & 0.976 (0.024) & 40.37 (4.52)$^*$ & 0.010 (0.013) & 0.966 (0.030)$^*$ & 38.50 (4.49)$^*$ & 0.020 (0.023) \\
    SwinIR & 0.973 (0.026)$^*$ & 39.71 (4.12)$^*$ & 0.018 (0.016)$^*$ & 0.953 (0.039)$^*$ & 36.61 (3.87)$^*$ & 0.044 (0.024)$^*$ \\
    \midrule
    MambaMIR & \textbf{0.980 (0.020)} & \textbf{41.59 (4.71)} & 0.011 (0.017) & \textbf{0.971 (0.026)} & \textbf{39.52 (4.58)} & 0.020 (0.026) \\
    MambaMIR-GAN & 0.978 (0.020) & 41.16 (4.75)$^*$ & \textbf{0.007 (0.015)} & 0.970 (0.025) & 39.19 (4.57)$^*$ & \textbf{0.010 (0.020)} \\
    \bottomrule
    \end{tabular}%
    }
  \label{tab:TABLE_COMP}%
\end{table*}%


\begin{figure*}[!t]
\begin{subfigure}
    \centering
    \includegraphics[width=\linewidth]{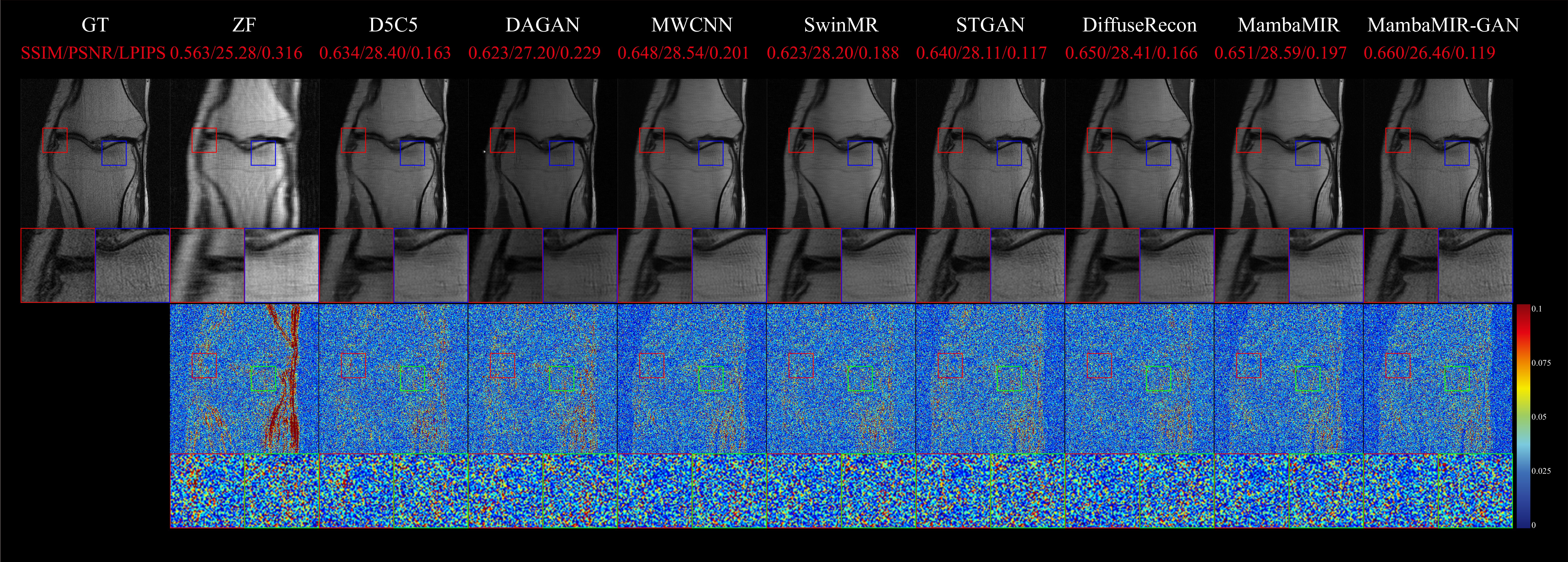}
    \caption{
    Visualised Results on FastMRI at AF $\times 4$. 
    Ground truth (GT), undersampled zero-filled (ZF) images, reconstruction results and corresponding error maps are presented.
    }
    \label{fig:FIG_VIS_MRI_FastMRI_AF4}
\end{subfigure}

\begin{subfigure}
    \centering
    \includegraphics[width=\linewidth]{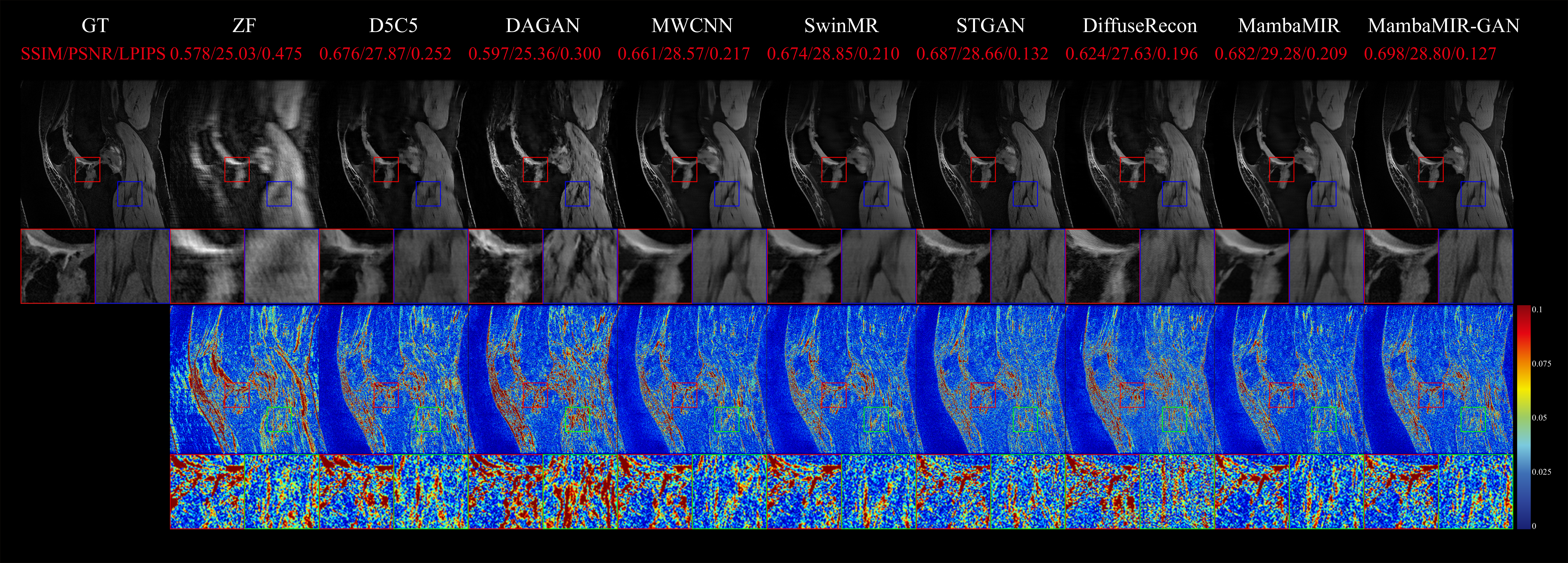}
    \caption{
    Visualised Results on SKMTEA at AF $\times 8$. 
    Ground truth (GT), undersampled zero-filled (ZF) images, reconstruction results and corresponding error maps are presented.
    }
    \label{fig:FIG_VIS_MRI_SKMTEA_AF8}
\end{subfigure}
\end{figure*}

\begin{figure*}[!t]
\begin{subfigure}
    \centering
    \includegraphics[width=\linewidth]{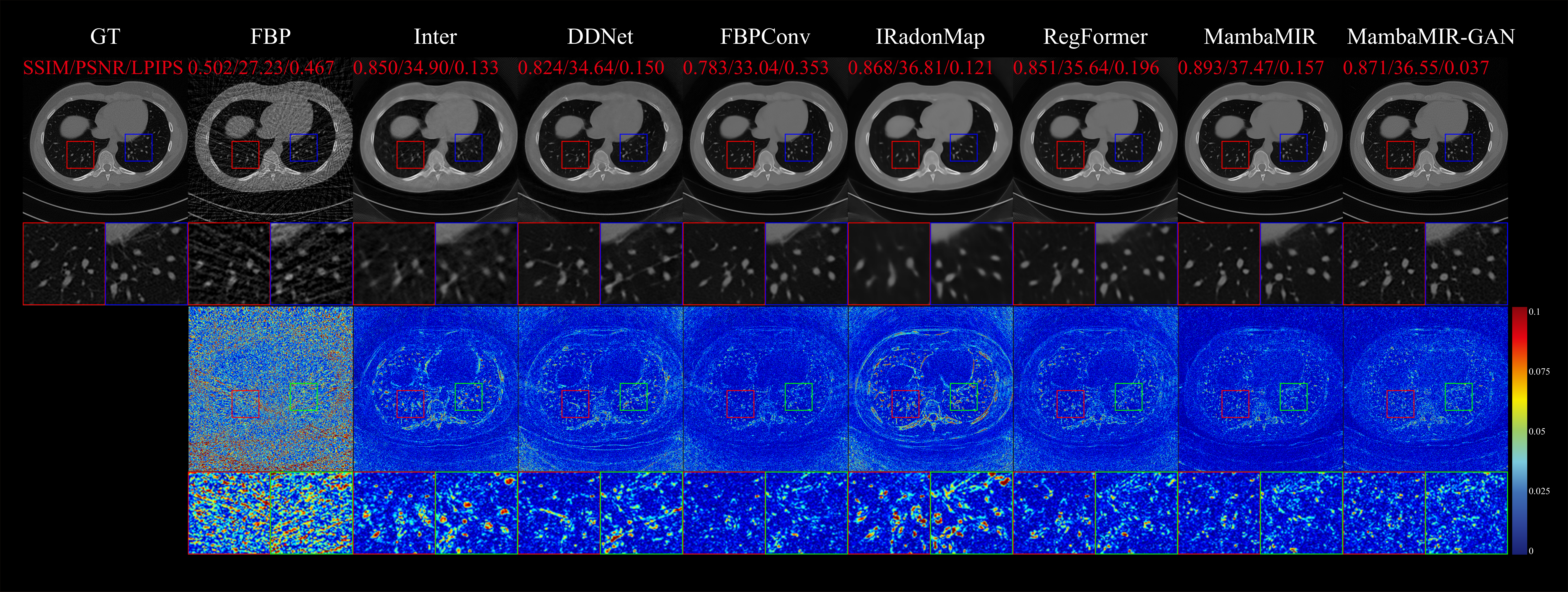}
    \caption{
    Visualised Results for SVCT on chest subset. 
    Ground truth (GT), sparse-view images reconstructed by Filtered Backprojection (FBP), reconstruction results and corresponding error maps are presented. CT images are normalised within the range of [-1024, 3096] HU for error map computation and display.
    }
    \label{fig:FIG_VIS_CT_CHEST}
\end{subfigure}

\begin{subfigure}
    \centering
    \includegraphics[width=\linewidth]{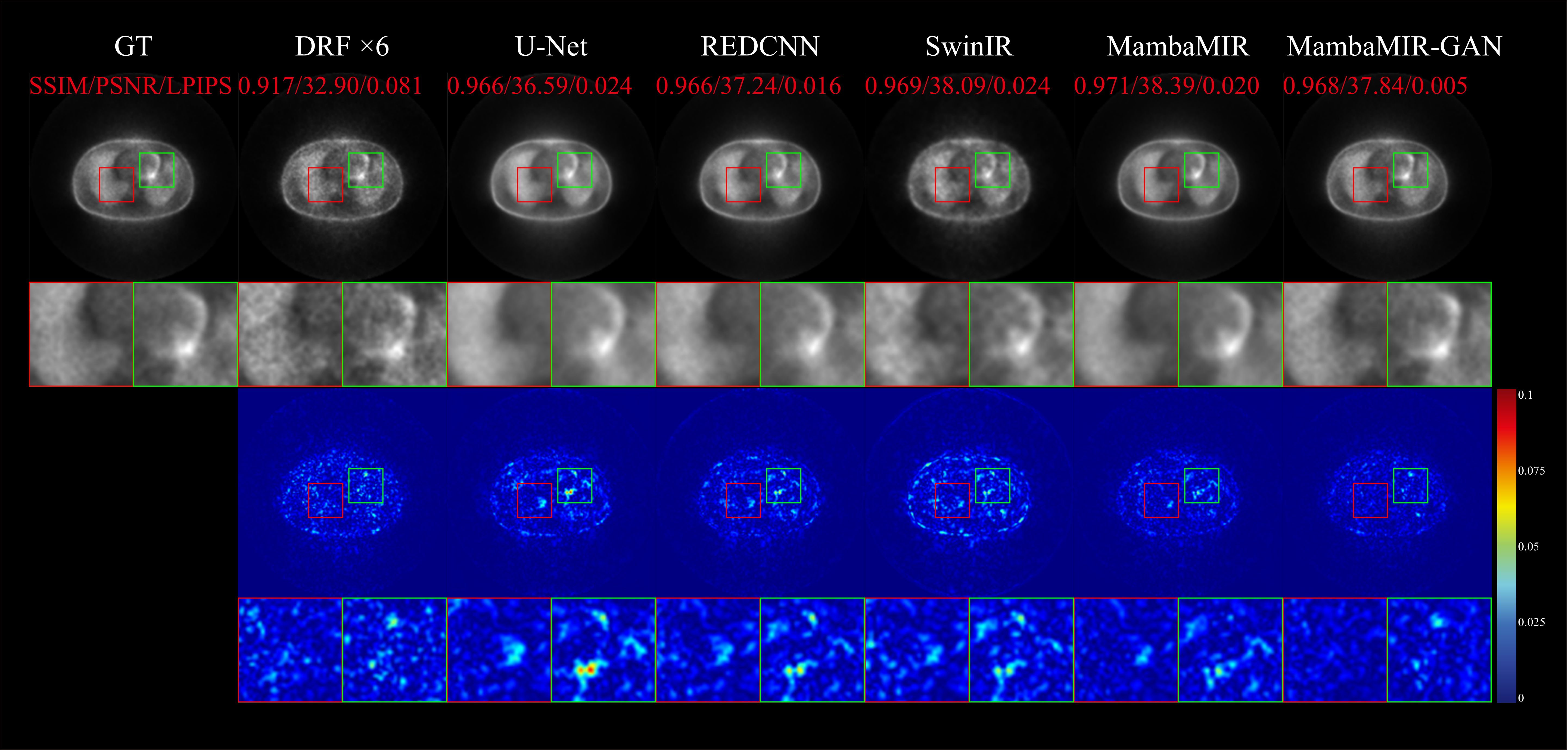}
    \caption{
    Visualised Results for LDPET at DRF $\times 6$. 
    Ground truth (GT), low-dose images, reconstruction results and corresponding error maps are presented.
    }
    \label{fig:FIG_VIS_PET_X6}
\end{subfigure}
\end{figure*}

\subsection{Comparisons with the SOTA}

In the experimental section, we compared our proposed MambaMIR and MambaMIR-GAN with baseline and SOTA methods for three different medical image reconstruction tasks, including fast MRI, SVCT and LDPET. 

Experiments for fast MRI included 
an unrolling-based method, D5C5~\citep{Schlemper2017D5C5}, 
an image enhancement-based and wavelet-coupled method, MWCNN~\citep{Liu2019Multi},
an image enhancement-based and Transformer-based method, SwinMR~\citep{Huang2022SwinMR}, 
GAN-based methods DAGAN~\citep{Yang2018DAGAN} and STGAN~\citep{Huang2022STGAN}, 
as well as a diffusion models-based method DiffuseRecon~\citep{Peng2022DiffuseRecon}.
Experiments were conducted on FastMRI at accelerate factors (AFs) of $\times4$ and $\times8$, and on SKMTEA at AFs of $\times8$ and $\times16$. 

For SVCT, image domain methods DDNet~\citep{zhang2018sparse} and FBPConv~\citep{jin2017deep}, a sinogram domain method View-Interpolation (Inter)~\citep{lee2017view}, a dual-domain method HDNet~\citep{hu2020hybrid}, a parameter-learnable inverse Radon transform IRadonMap~\citep{he2020radon}, and an unfolding method RegFormer~\citep{xia2023transformer} were included.
Results on abdomen and chest subsets were reconstructed from 60-views sinograms with uniform sampling. 

For LDPET, we compared the proposed method with 
image enhancement-based methods U-Net~\citep{Ronneberger2015UNet} and REDCNN~\citep{chen2017low}, 
as well as a Transformer-based SwinIR~\citep{Liang2021SwinIR}.
Quantitative results with DRF$\times3$ and DRF$\times6$ are presented.

All quantitative results for comparison studies can be found in Table~\ref{tab:TABLE_COMP}.
Visualised reconstruction samples can be found in Fig.~\ref{fig:FIG_VIS_MRI_FastMRI_AF4} for FastMRI at AF$\times4$, Fig.~\ref{fig:FIG_VIS_MRI_SKMTEA_AF8} for SKMTEA at AF$\times8$, Fig.~\ref{fig:FIG_VIS_CT_CHEST} for SVCT on chest scans, and Fig.~\ref{fig:FIG_VIS_PET_X6} for LDPET at DRF$\times6$.
According to the results, our proposed MambaMIR and MambaMIR-GAN achieve comparable results or outperform current SOTA methods, where MambaMIR tends to provide results with better reconstruction fidelity, while MambaMIR-GAN presents results with superior perceptual experience.

\subsection{Ablation Studies}

\subsubsection{Component Validity}

To evaluate the validity of each component in the network architecture, ablation experiments, with the removal of a single component in each run, were performed on the FastMRI at AF $\times8$ with a patch size of 2. 

According to TABLE~\ref{tab:TABLE_ABL_COMP}, we found that the utilisation of all components has a positive impact on the reconstruction results. 
Among them, the use of wavelet significantly improved the model performance with only a slight effect on the model size. The use of MLPs in AMSS blocks and self-attention modules in the bottleneck, as two common designs in the network, was also shown to be effective.

\subsubsection{Hyperparameter}

\begin{table}[htbp]
  \centering
  \caption{
    Ablation studies for model component validity conducted on FastMRI at AF $\times 8$.
    with a patch size of 2. 
    Structural Similarity Index Measure (SSIM) and the number of parameter (\#PARAMs) are reported to reflect the performance and model size.
    `FULL': standard MambaMIR;
    WAMSS: Wavelet-embedded Arbitrary-Masked State Space Blocks;
    WDown/WUp: Wavelet-based downsampling/upsampling modules;
    MLP: Multilayer perceptrons in AMSS Blocks;
    Attn.: multi-head self-attention modules in the bottleneck.
    }
    \begin{tabular}{lcc}
    \toprule
    \multicolumn{1}{c}{Settings} & SSIM  & \#PARAMs (M) \\
    \midrule
    FULL  & 0.5741 & 50.227 \\
    WAMSS $\rightarrow$ AMSS & -0.0017 & -0.000 \\
    WDown/WUp $\rightarrow$ Down/Up & -0.0020 & -1.328 \\
    w/o Wavelet (Two Lines Above) & -0.0027 & -1.328 \\
    w/o MLP & -0.0026 & -9.175 \\
    w/o Attn. & -0.0017 & -0.264 \\
    \bottomrule
    \end{tabular}
  \label{tab:TABLE_ABL_COMP}%
\end{table}%

Ablation studies in the hyperparameter setting were conducted and presented in Fig.~\ref{fig:FIG_PLOT_ABL_HYPER} (A), exploring the patch size, the resolution of random cropping during training, and the number of S6's latent space channels. 

Regarding the patch size, both the reconstruction performance (SSIM) and the computational cost (GFLOPs) increase as the patch size gets smaller. We can observe that the computational complexity of MambaMIR approximately increases linearly with the length of the sequence, consistent with theoretical predictions (patch size: $\times 2$, sequence length: $\times 4$, FLOPs: around $\times 4$). In our experiments, MambaMIR with patch size of 1 is applied for benchmarking for a fair comparison, while MambaMIR with patch size of 2 was used for ablation studies due to hardware limitation. 

In terms of the resolution of random cropping during training, the SSIM initially improves, reaching an optimal value before subsequently declining as the resolution increases, meanwhile, the FLOPs grow as the resolution increases. The optimal value exists since random cropping during training can be regarded as a data augmentation. We choose a resolution of $192 \times 192$ during training as a trade-off between performance and computational complexity. 

For the number of latent space channels (\#Channel) in S6,  both the reconstruction performance (SSIM) and the computational cost (FLOPs) increase as \#Channel increasing. As a trade-off between performance and computational complexity, \#Channel is set to 128.

\subsection{Transformer v.s. Mamba}

Mamba is regarded as a powerful competitor of Transformer. In this section, we further explored the comparison between our proposed MambaMIR and the Transformer-based counterparts.

To conduct a fair comparison, we replaced the AMS6 block in our proposed MambaMIR with multi-head self-attention (MSA) or shifted-window MSA (SWMSA), while preserving the rest components.
We use ``SxTy'' to indicate the Transformer-based counterpart, where x is the number of SWMSA and y is the number of MSA in both the encoder and the decoder. Since MSA has much more computational complexity than SWMSA, SWMSA is always applied in the shallower stage in our U-shape architecture. As reported in Fig.~\ref{fig:FIG_PLOT_ABL_HYPER} (B), the proposed MambaMIR outperforms all different Transformer-based counterparts and achieves a better Perception-Distortion Trade-off~\citep{Blau2018Perception}, while yielding a reasonable computational complexity. 

\begin{figure*}[!t]
    \centering
    \includegraphics[width=\linewidth]{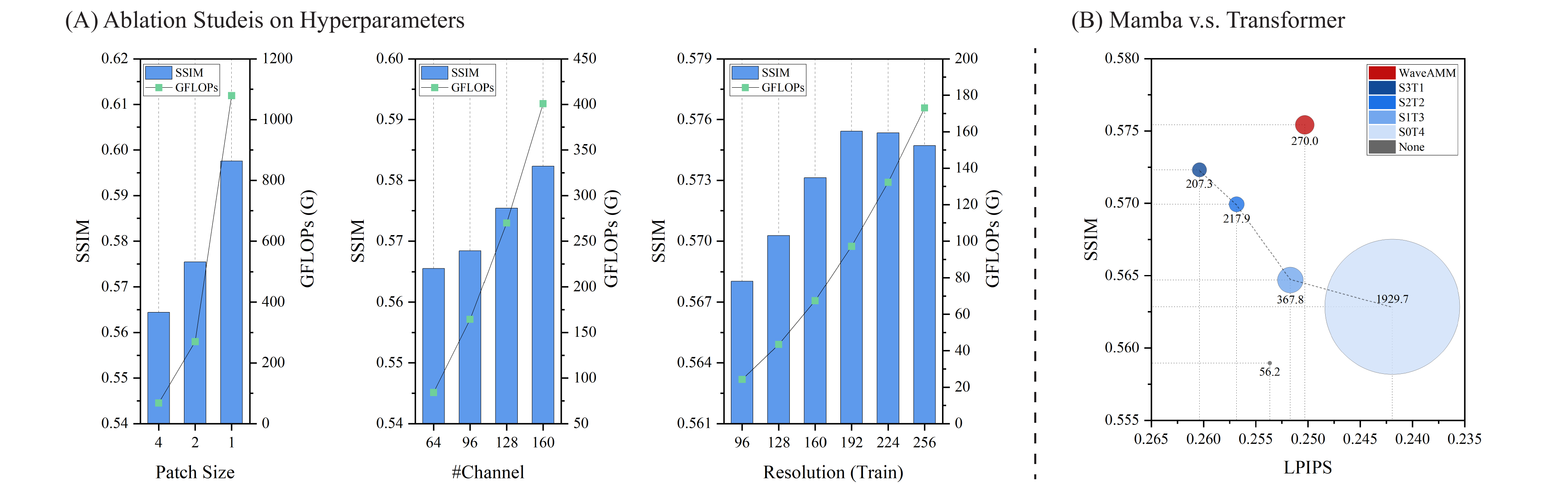}
    \caption{
    (A) Ablation studies on hyperparameters regarding the patch size, the randomly cropping resolution during training, and the number of S6's latent space channels (\#Channel);
    (B) Experiments between Mamba-based and Transformer-based models.
    The size of the data circle and the number below indicate the computational complexity (GFLOPs). 
    }
    \label{fig:FIG_PLOT_ABL_HYPER}
\end{figure*}

\section{Discussion}

In this paper, we have proposed MambaMIR, an innovative Mamba-based model, along with its GAN-based variant, MambaMIR-GAN, for joint medical image reconstruction and uncertainty estimation. 
According to TABLE~\ref{tab:TABLE_COMP}, MambaMIR-GAN tends to achieve the best LPIPS score and MambaMIR achieve the best PSNR score. 
It can be observed from the visualised example (Fig.~\ref{fig:FIG_VIS_CT_CHEST}, Fig.~\ref{fig:FIG_VIS_MRI_SKMTEA_AF8} and Fig.~\ref{fig:FIG_VIS_PET_X6}) that MambaMIR-GAN tends to produce more detailed texture information, while MambaMIR provides smoother recontruction to ensure reconstruction fidelity, which satisfies a Perception-Distortion Trade-off~\citep{Blau2018Perception}.
The experimental results have suggested that both MambaMIR and MambaMIR-GAN delivered superior performance in medical image reconstruction. In particular, MambaMIR tends to provide results with better reconstruction fidelity, while MambaMIR-GAN may provide reconstructions that better align with human perceptual qualities. 

An essential advantage of Mamba is the global sensitivity alongwith linear complexity. As Fig.~\ref{fig:FIG_VIS_ERF} illustrated, our MambaMIR yield a larger Effective Receptive Fields~\citep{Luo2016Understanding} compared to comparison methods, demonstrating its superior global sensitivity and long-range dependency.

In addition to the outstanding reconstruction results, our proposed MambaMIR can provide uncertainty maps by repeat sampling. These maps have visually represented the model's confidence in the reconstructed images, highlighting areas of potential uncertainty, which may signal regions with lower image quality or artefacts. 
As Fig.~\ref{fig:FIG_VIS_UNC} (A) illustrated, the high-uncertainty area for MRI knee reconstruction is located mainly in tissue with informative details (high-frequency area). 
For CT chest reconstruction, the high-uncertainty area indicated that edges of tissue and bones are with high-uncertainty.
For PET reconstruction, it can be observed that areas with higher radioactive concentrations demonstrate increased uncertainty. 
According to Fig.~\ref{fig:FIG_PLOT_Dropout} and Fig.~\ref{fig:FIG_VIS_UNC} (B), experimental results have shown that dropout leads to a consistent performance drop on different datasets with different dropout rates, the severity of which is positively correlated with the dropout rate. Compared to MC dropout, our MC-ASM mitigates the performance drop while providing reasonable uncertainty maps without the need for hyperparameter tuning.

\begin{figure}[ht]
\begin{subfigure}
    \centering
    \includegraphics[width=\linewidth]{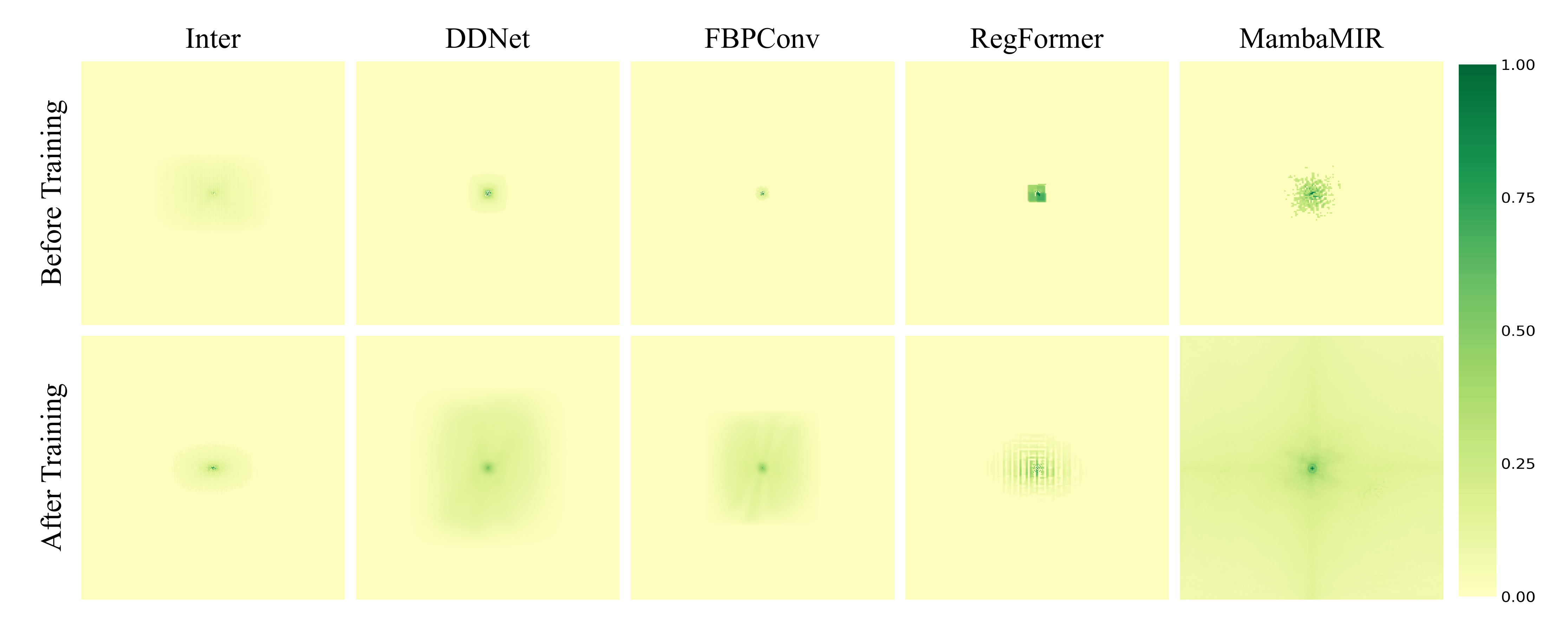}
    \caption{Comparison of Effective Receptive Fields before and after training between the proposed MambaMIR and other methods on for SVCT on abdomen subset.}
    \label{fig:FIG_VIS_ERF}
\end{subfigure}

\begin{subfigure}
    \centering
    \includegraphics[width=\linewidth]{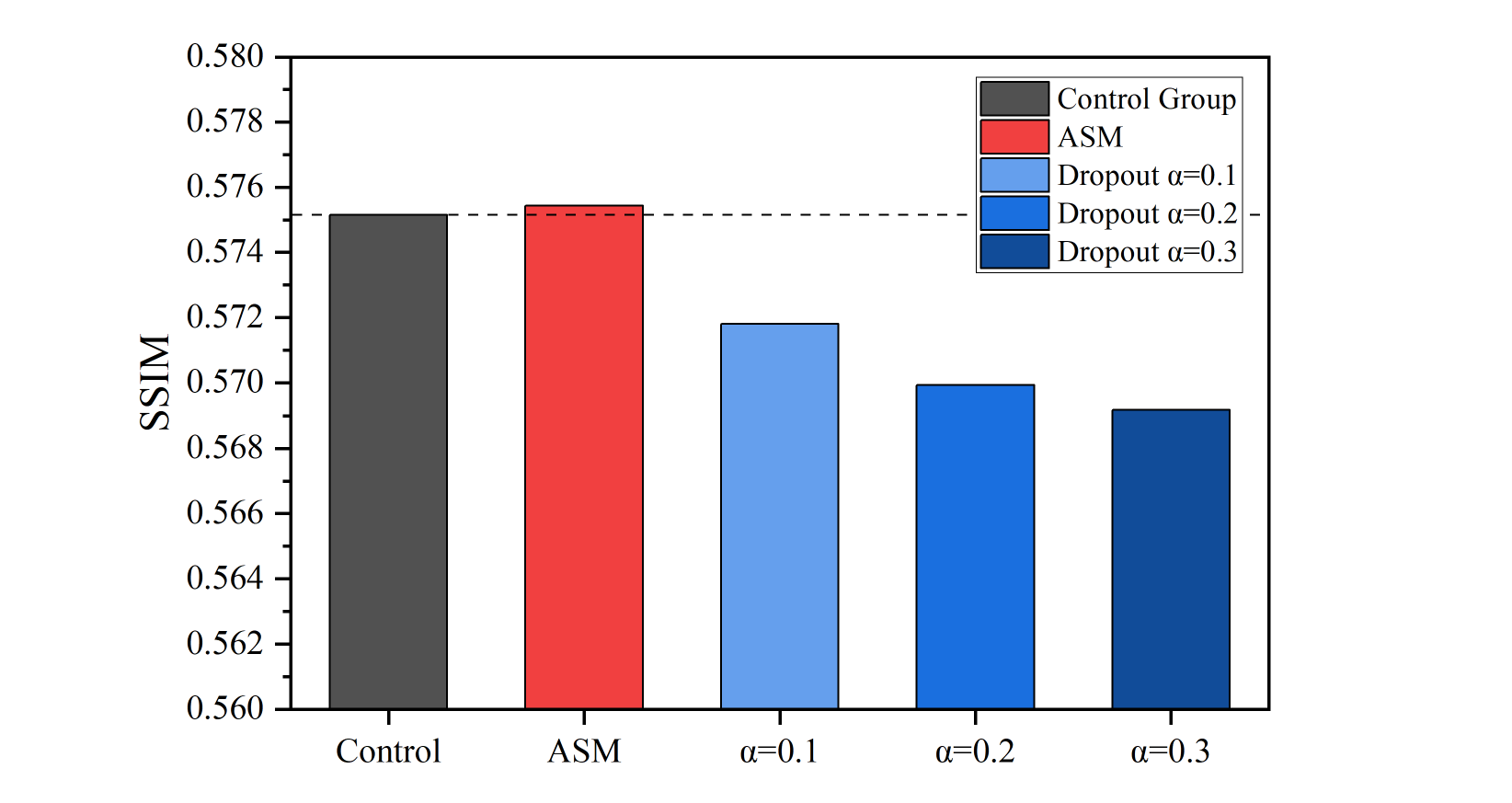}
    \caption{
    Quantitative comparison on FastMRI dataset between 1) MambaMIR without MC-ASM or MC dropout (control group), 2) MambaMIR with MC-ASM and 3) MambaMIR with MC Dropout using differet dropout rate.
    }
    \label{fig:FIG_PLOT_Dropout}
\end{subfigure}
\end{figure}

\begin{figure*}[!t]
    \centering
    \includegraphics[width=\linewidth]{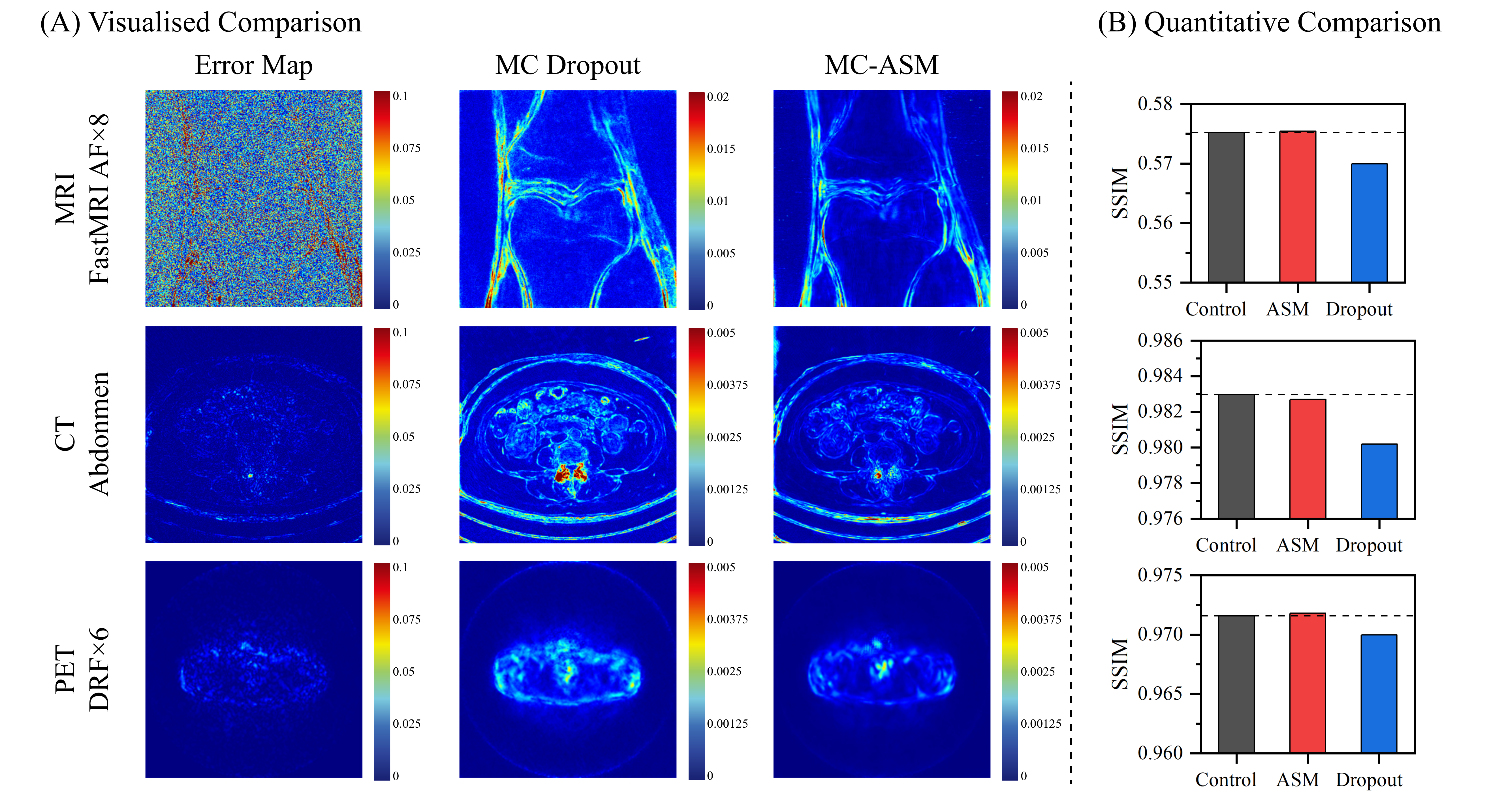}
    \caption{
    (A) Visualised samples of uncertainty maps provided by MC dropout ($\alpha=0.2$) and our MC-ASM, along with the corresponding error maps.
    (B) Quantitative comparison between 1) MambaMIR without MC-ASM or MC dropout (control group), 2) MambaMIR with MC-ASM and 3) MambaMIR with MC Dropout ($\alpha=0.2$) on three datasets.
    }
    \label{fig:FIG_VIS_UNC}
\end{figure*}

Comprehensive ablation experiments have been performed to evaluate the validity of model components. 
The use of Wavelet (both WDown/WUp and WAMSS) has shown significant beneficial to reconstruction with few network size increases.
Original VMamba~\citep{Liu2024VMamba} discarded the paradigm of \textit{Norm $\rightarrow$ Attention $\rightarrow$ Norm $\rightarrow$ MLP} from Vision Transformer~\citep{Dosovitskiy2020ViT}, instead they only used \textit{Norm $\rightarrow$ Vision State Space Block} for lighter network architecture. 
According to Fig.~\ref{fig:FIG_PLOT_ABL_HYPER}, ablation studies have shown that it is still necessary to retain MLP in our MambaMIR for medical image reconstruction, although it largely increases network size. 
In addition, the utilisation of self-attention module in the bottleneck effectively has improved the performance without significantly increasing the model size.

The choice of hyperparameter across the patch size, the training resolution, as well as the number of deep embedding channels, have also been evaluated and presented in Fig.~\ref{fig:FIG_PLOT_ABL_HYPER} (A).
Patch size is an essential parameter for low-level tasks such as medical image reconstruction.
Smaller patches can capture finer details since they focus on smaller regions, ensuring better reconstruction fidelity. However, the computational cost increases significantly when the number of patches to process multiplies~\citep{Huang2022SwinMR}. 
Our ablation studies have shown a similar trend in the relationship between the choice of patch size, the resulting reconstruction performance, and the computational cost. 
Moreover, it can be observed that the computational complexity of MambaMIR approximately increases linearly with the length of the sequence, which is consistent well with theoretical predictions of Mamba's complexity. 

In terms of the number of latent space channels (\#Channel) in S6, typically, for low-level tasks, the performance increases and finally becomes stable (or drops) as \#Channel increases. 
A larger latent space can provide a model with a higher capacity to encode complicated features and details of the input images, which generally allows for more detailed and accurate reconstructions, especially for complex images with a lot of variance. However, models with a very high-dimensional latent space can risk overfitting, particularly if the training data is limited or not diverse enough, which leads to a performance drop.
Ablation studies have shown that both the model performance and the corresponding computation cost increase as \#Channel increases, however, no performance drop has been observed when \#Channel is large. This reflects that our proposed MambaMIR has the potential to have better performance with higher-dimensional latent space.

Random cropping during model training is a common data augmentation technique that can significantly affect the final results of the model. Typically, random cropping helps to prevent overfitting by ensuring that the model does not learn to rely on specific features located in particular parts of an image. While random cropping increases robustness and generalisation, it might also lead to a loss of important contextual information when the cropping region is too small. 
Ablation studies have shown that the reconstruction performance initially improves, reaching an optimal value before subsequently declining as the training resolution increases, where the optimal performance is a balance of data augmentation and information richness for a single crop.

To further explore the potential of the Mamba-based model as a competitor of the Transformer, a fair comparison between our MambaMIR and its Transformer counterparts has been conducted. 
As illustrated in Fig.~\ref{fig:FIG_PLOT_ABL_HYPER} (B), it can be observed that applying self-attention (without the window mechanism) in the shallow and high-resolution stage is extremely computationally expensive due to its quadratic complexity. From `S3T1' to `S0T4', more self-attention modules are applied, which leads to an extensive growth of computational complexity, meanwhile pushing the balance from better fidelity to better perception following a Perception-Distortion Trade-off~\citep{Blau2018Perception}. 
The proposed MambaMIR has shown superiority over all different Transformer-based counterparts in terms of two trade-offs: Perception-Distortion Trade-off and Performance-Complexity Trade-off.

\section{Conclusion}
\label{sec:conclusion}
In conclusion, our proposed MambaMIR and MambaMIR-GAN represent significant advances in the field of medical image reconstruction.
The proposed generalised framework has been achieved superior performance on fast MRI, SVCT and LDPET, which proves its scalability and potential for other reconstruction applications such as ultrasound or low-dose CT reconstruction.
The proposed MC-ASM mechanism provides reliable uncertainty estimation without the need for hyperparameter tuning and mitigates performance drop. 

Future studies may investigate the scalability of these models for various imaging modalities and their potential for computational efficiency.

\section*{Acknowledgments}
This study was supported in part by the ERC IMI (101005122), the H2020 (952172), the MRC (MC/PC/21013), the Royal Society (IEC/NSFC/211235), the NVIDIA Academic Hardware Grant Program, the SABER project supported by Boehringer Ingelheim Ltd, NIHR Imperial Biomedical Research Centre (RDA01), Wellcome Leap Dynamic Resilience, and the UKRI Future Leaders Fellowship (MR/V023799/1). 





\bibliographystyle{model2-names.bst}\biboptions{authoryear}
\bibliography{refs}



\end{document}